%% file: main.tex
\documentclass[journal]{IEEEtran}

\input{config/packages}

\input{config/macros}

\begin{document}

\title{Identification of $dq$-Asymmetric Impedances as Complex Transfer Functions Using a Single Arbitrary Excitation}
\author{Mohamed~Abdalmoaty,~\IEEEmembership{Member,~IEEE,}
        Zheran~Zeng,~\IEEEmembership{Graduate Student Member,~IEEE,}\\
        Dongsheng~Yang,~\IEEEmembership{Senior Member,~IEEE,}
        and~Florian~D\"orfler,~\IEEEmembership{Fellow,~IEEE}%
\thanks{M.~Abdalmoaty and F.~D\"orfler are with the Automatic Control Laboratory,
ETH Zurich, 8092 Zurich, Switzerland (e-mail: mabdalmoaty@ethz.ch).}%
\thanks{Z.~Zeng and D.~Yang are with the Department of Electrical Engineering, Eindhoven
University of Technology, Eindhoven, The Netherlands.}%
\thanks{A preliminary part of the methods underlying this work was presented at the 24th Power Systems
Computation Conference \cite{Abdalmoaty2027epsr}.}}

\maketitle

\begin{abstract}
Cross-coupling between the $dq$ coordinates makes the identification of asymmetric grid
impedances a challenging problem, particularly near the fundamental frequency
where the asymmetric coupling is strongest. Existing schemes usually handle it either by
perturbing the two coordinates sequentially, which lengthens the measurement, or by
using a time-domain method with a global parametric model whose order must be tuned. This paper
develops an active non-parametric frequency-domain method that avoids both. The
equivalent impedance is parameterized by a pair of single-input single-output
complex transfer functions. Each spectral line is fitted with a local
rational model; the leakage and transient contributions are estimated, so that neither periodic steady-state
excitation nor repeated excitation cycles are required. We give the exact finite-time discrete
Fourier transform relation for the conjugate-coupled complex-signal model, and 
analyse the distortion that a stationary-frame filter
placed ahead of the Park transform imposes on the identified pair. The method is
validated on a controller hardware-in-the-loop platform against an analytically
derived small-signal model, for a symmetric grid and for the same grid
with an added grid-following converter that renders it asymmetric. Both complex transfer functions and all 
four real transfer functions of the $dq$ impedance are
recovered over a wide band from a single one-second record of a random
excitation, at \SI{1}{\Hz}  resolution.
\end{abstract}


\input{sections/introduction}
\input{sections/formulation}
\input{sections/identification}
\input{sections/measurement}
\input{sections/hil}
\input{sections/conclusions}

\appendices
\input{sections/appendix}

\bibliographystyle{IEEEtran}
\bibliography{config/IEEEabrv,refs}

\end{document}

%% file: config/packages.tex
\usepackage{comment}

\usepackage{graphicx}
\usepackage{epstopdf}

\newif\ifshowtodos
\showtodostrue
\ifshowtodos
  \usepackage[textsize=footnotesize]{todonotes}
\else
  \usepackage[disable]{todonotes}
\fi

\usepackage{cite}

\usepackage[dvipsnames]{xcolor}
\usepackage{pgfplots}
\usetikzlibrary{backgrounds}
\usepgfplotslibrary{colorbrewer}
\usepgfplotslibrary{groupplots}
\pgfplotsset{compat = 1.16, cycle list/Set1-8}
\usetikzlibrary{pgfplots.statistics}
\usepackage{pgfplotstable}

\usetikzlibrary{external}
\tikzset{external/force remake=false}
\tikzexternaldisable

\pgfplotsset{
    boxplot/draw/average/.code={%
        \color{.!0!black}
        \draw[/pgfplots/boxplot/every average/.try]
            \pgfextra
            \pgftransformshift{%
                \pgfplotsboxplotpointabbox
                    {\pgfplotsboxplotvalue{average}}
                    {0.5}%
            }%
            \pgfuseplotmark{\tikz@plot@mark}%
            \endpgfextra
        ;
    },
}

\definecolor{uugreen}{HTML}{14B03D}
\definecolor{uured}{RGB}{191,45,56}
\definecolor{uublue}{RGB}{0, 106, 178}
\definecolor{backgroundcolor}{rgb}{0.08, 0.38, 0.74}
\definecolor{greenBackgroundcolor}{rgb}{0.4660, 0.6740, 0.1880}
\definecolor{backgroundcolor2}{rgb}{0, 0.4470, 0.7410}
\definecolor{orangecolor}{rgb}{0.85, 0.325, 0.098}

\usetikzlibrary{circuits.ee.IEC}
\usetikzlibrary{arrows}
\tikzstyle{roundnode} =[circle, draw=blue!60, fill=blue!5, scale = 0.5]
\usetikzlibrary{positioning}
\usetikzlibrary{arrows.meta}

\usepackage{environ}
\makeatletter
\newsavebox{\measure@tikzpicture}
\NewEnviron{scaletikzpicturetowidth}[1]{%
  \def\tikz@width{#1}%
  \def\tikzscale{1}\begin{lrbox}{\measure@tikzpicture}%
  \BODY
  \end{lrbox}%
  \pgfmathparse{#1/\wd\measure@tikzpicture}%
  \edef\tikzscale{\pgfmathresult}%
  \BODY
}
\makeatother

\usepackage{dblfloatfix}

\usepackage[cmex10]{amsmath}
\usepackage{amssymb}
\usepackage{amsthm}
\theoremstyle{definition}
\newtheorem{theorem}{\bf Theorem}
\newtheorem{lemma}{\bf Lemma}
\newtheorem{proposition}{\bf Proposition}

\newtheorem{remark}{\bf Remark}

\usepackage{siunitx}
\usepackage{booktabs}
\usepackage{makecell}
\usepackage{multirow}
\usepackage{multicol}

\usepackage[T1]{fontenc}
\usepackage{newtxtext}
\usepackage[caption=false,font=footnotesize]{subfig}
\usepackage{stackengine}
\usepackage[bottom]{footmisc}
\usepackage[hidelinks]{hyperref}

\pgfplotsset{
  hilaxis/.style={
    grid=major, grid style={gray!15},
    minor tick num=1,
    mark size=1.5pt,
    line width=.05pt,
    yticklabel style={font=\scriptsize},
    xticklabel style={font=\scriptsize},
    ylabel style={font=\scriptsize, yshift=-0.2cm},
    xlabel style={font=\scriptsize, yshift=0.1cm},
    title style={font=\small},
    legend style={font=\footnotesize, draw=none, fill=none,
                  at={(0.5,1.03)}, anchor=south, legend columns=2,
                  column sep=0.3ex},
    scaled x ticks=false, scaled y ticks=false,
  },
  truecurve/.style={black!90, line width=1pt, mark=none},
  estmarks/.style={backgroundcolor, only marks, mark=x, line width=0.5pt},
}

%% file: sections/introduction.tex
\section{Introduction}\label{sec:intro}

\IEEEPARstart{T}{he} displacement of synchronous generation by
converter-interfaced resources has left modern grids with dynamics that are both
faster and far less well documented than those of the systems they replace.
Subsystem interactions can vary with operating point and fixed analytical models
quickly cease to describe the system \cite{Wang2019}. The difficulty is
aggravated by how little model information is shared between stakeholders:
manufacturers treat converter models as proprietary, network operators
frequently hold only steady-state or aggregated representations and are reluctant to share even these, and large flexible loads such as data centers contribute volatile dynamics that are so far not well characterized.

Measuring the impedance using data-driven methods is a promising alternative
\cite{DeMeerendre2020}; it yields a model of the local dynamics without
requiring disclosure of protected information or the knowledge of what devices
are actually connected behind the point of interest. The uses of such a model
are broad. Evaluated at the fundamental frequency $\omega_g$ it informs 
voltage-stability margins, transfer limits and grid-strength indicators, and when
characterized across a band it supports harmonic studies, filter and controller
design, the diagnosis of subsynchronous and inter-area phenomena, and the tuning
of grid-connected converter controls
\cite{Harnefors2007a,cespedes2014adaptive,Wang2014,Wang2018}. Since the
interaction between converters and the grid can degrade power quality and even
trigger instabilities \cite{Mollerstedt2000,Liserre2006,li2017unstable},
impedance-based criteria have become a standard tool for small-signal
stability assessment \cite{belkhayat1997,sun2009small,Rygg2016,Chen2024}. In a
three-phase system the object of interest, $Z_g(s)$, is a multi-input
multi-output transfer function (TF) between small-signal terminal
voltages and currents at the point of common coupling (PCC), most often written
in a synchronous $dq$ frame. \vspace{-0.7cm}

\subsection{Identification methods and their limitations}

Two decades of work, mostly in the power-electronics literature, have produced a
large number of identification schemes  \cite{Stiegler2015,DeMeerendre2020}.
They can be categorized into \emph{passive}, \emph{quasi-passive} and
\emph{active} families. \emph{Passive} schemes exploit the harmonic distortion
already present at the PCC \cite{Gu2012,Hoffmann2014}. Because the distortion
usually sits at only a handful of harmonics, the estimate is available only on a
sparse set of frequencies --- adequate for tracking a short-circuit level, but
not for the wideband characterization that stability analysis or control design
requires. \emph{Quasi-passive} schemes pair a triggering mechanism
\cite{Garcia2014,Cobreces2009} with an active measurement so that the grid is
disturbed only when a change is detected; they inherit whatever limitations that
measurement has. \emph{Active} schemes create the excitation deliberately,
either by switching resistive or capacitive loads \cite{Girgis1989,Jordan2018}
or by injecting perturbations through dedicated equipment
\cite{francis2011algorithm,huang2009small,Rygg2016}: frequency sweeps
\cite{huang2009small} step a sinusoid through the band, while wideband injection \cite{Rygg2016} excites many
frequencies at once. Both evaluate the measurement by steady-state Fourier
analysis and require purpose-built hardware and long records. This limits their application to commissioning and laboratory use.

A more practical approach is to use existing grid-connected converters to
excite the grid, by superimposing an excitation signal on an
inner control loop reference.
Impulse excitation \cite{cespedes2012online,liu2020analysis} achieves a very
short perturbation, but concentrates the injected energy into a
brief, large transient that can degrade power quality, depart from linear behavior, or trip protection devices. Lower-amplitude maximum-length binary
sequences (MLBS) avoid that
\cite{martin2013wide,riccobono2017noninvasive,luhtala2018implementation,roinila2017mimo},
at the cost of a structural restriction: either the impedance is assumed
$dq$-symmetric and the cross-coupling effects are neglected (not
justified in a converter-dominated grid), or two linearly independent
perturbations are applied one after the other, which doubles the record and
requires the grid to be unchanged between them.

Fast impedance measurement is crucial in practical power grids. A small-signal model is obtained by linearization at a single operating point, yet real grids rarely hold a constant operating point for long, due to the intermittency of renewable generation. The concern is most acute at low frequencies, where measurement via classical techniques may require several periods of a periodic steady-state excitation. 
There have been a few works trying to address this. For example, the two sequential
perturbations are avoided in
\cite{Berg2022} by injecting two orthogonal binary sequences at once, an 
MLBS on one axis and the inverse-repeat sequence
derived from it on the other. The  two sequences
excite interleaved subsets of the frequency grid, so an interpolation step is
needed to bring the second channel onto the frequencies of the first. Because
the inverse-repeat sequence is twice as long as the sequence it is derived
from, a data record delivering a frequency resolution $f_\mathrm{res}$ must span
$2/f_\mathrm{res}$, the same duration as the two sequential
perturbations.  This ensures that both axes are measured under
identical operating conditions. Because the analysis uses a steady-state discrete
 Fourier transform (DFT), the record needs to be long enough, usually containing
 an integer number of periods.  A different route is taken in \cite{Haberle2023}, using discrete-time
auto-regressive exogenous (ARX) models \cite{ljung1998system}. These need no
special excitation and tolerate non-periodic data. Similarly, a time-domain method was proposed
in \cite{gong2020dq} for converter identification. However, these methods require a
global parametric model whose order must be selected carefully, and their accuracy
degrades markedly when that choice is wrong.

A different concern, common to all methods, is that synchronous-$dq$
identification requires the phase angle of the PCC voltage. This angle is usually
obtained from a phase-locked loop (PLL), whose dynamics may bias the result.
The bias can be corrected via high-pass filtering \cite{gong2018impact} or using a known PLL model
\cite{shen2013analysis}. 
Here it is avoided altogether,
since at a constant grid frequency the fundamental phase is needed only as a
static reference.

\vspace{-0.2cm}

\subsection{Contributions}
This paper develops an active method for
$dq$-asymmetric grid models requiring neither sequential nor
steady-state perturbation. The key contributions are summarized as:

\begin{enumerate} 
\item A {single-record} identification framework is developed for 
 $dq$-asymmetric impedances: random excitation is applied once and the response is captured in one triggered acquisition, eliminating the need for sequential injections of 
linearly independent perturbation signals.

\item An exact finite-record DFT relation is established for arbitrary 
non-periodic data. The transient term associated with finite-record truncation
 and spectral leakage is incorporated into the  model, thereby removing
 the requirement for periodic steady-state measurement.

\item A complex transfer function formulation separates the direct and conjugate-coupling dynamics,
 providing an explicit characterization of $dq$-asymmetry. 
This representation  enables an algebraically compact estimation of the
 dynamics, the transient term,  and straightforward correction of distortions caused by stationary-frame filtering ahead of the Park transform.

\end{enumerate}

Note that the grid is  still assumed to operate in a steady state about a fixed stable operating point, as in any
small-signal formulation; what is not required is that the \emph{excitation} be
in periodic steady state while measuring the system's response.  The record length is set only by the desired frequency
resolution and the signal-to-noise ratio (SNR).

The remainder of the paper is organized as follows.
Section~\ref{sec:formulation} describes the system and formulates the
identification problem. Section~\ref{sec:analysis} presents the proposed
method: the exact finite-time DFT relation, the local rational model estimator, and its
identifiability and excitation requirements. Section~\ref{sec:measurement}
analyses the measurement chain, in particular the distortion imposed by an
anti-alias/decimation filter placed ahead of the Park transform, and the choice
of reference frame. Section~\ref{sec:hil} reports the experimental validation
on a controller hardware-in-the-loop platform, and Section~\ref{sec:conclusions}
concludes the paper.

%% file: sections/formulation.tex
\section{System Description}\label{sec:formulation}

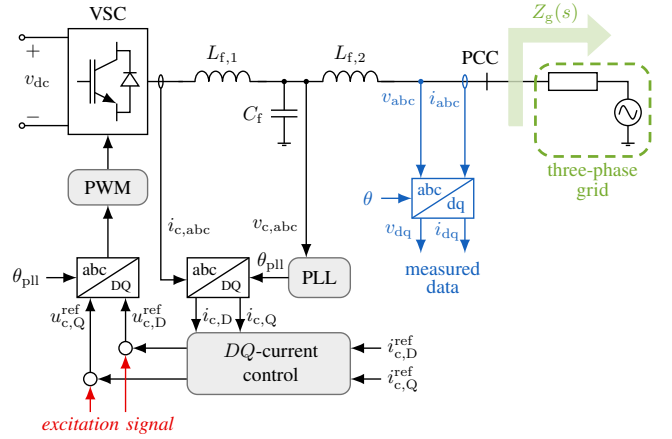
\begin{figure}[t]
    \centering 
    \input{figures/inverter_scheme_small}
        \vspace{-0.8cm}
    \caption{Grid-connected converter system with excitation in the control loop.}
    \label{figure:fig1}
\end{figure}

\subsection{Small-signal model}

The quantity to be identified is the dynamic small-signal Th\'{e}venin
equivalent impedance of an AC three-wire, three-phase grid, reconstructed from
time-domain samples of the terminal voltages $v_{abc}$ and currents $i_{abc}$ at
a PCC of interest; see Fig.~\ref{figure:fig1}. Nothing is assumed about what
lies behind the PCC, which may contain machines, passive loads and
actively controlled power-electronic devices in any combination. Assuming balanced
operation, Park's transformation evaluated at the steady-state grid
frequency $\omega_g$ carries the phase quantities into constant $dq$-coordinates
and thereby supplies a constant  operating point about which the system is linearized.

The continuous-time small-signal model is the set of four real SISO transfer operators relating
the $dq$ small-signal currents to the corresponding voltages, \vspace{-0.1cm}
\[
\begin{bmatrix}
    \Delta v_d(t)\\
    \Delta  v_q(t)
\end{bmatrix} =
\overbrace{\begin{bmatrix}
    Z_{dd}(\p) & Z_{dq}(\p)\\
    Z_{qd}(\p) & Z_{qq}(\p)\\
\end{bmatrix}}^{=:Z_g(\p)}
\begin{bmatrix}
    \Delta  i_d(t)\\
    \Delta  i_q(t)
\end{bmatrix},
\]
in which $Z_g(\p)$ is a real $2\times2$ array of transfer operators,
$\p\!=\!\frac{\diff}{\diff t}$ is the derivative operator, and $\Delta$ marks the small deviation from the
steady-state value. In practice, the small-signal quantities are obtained by
removing the mean of the time series, or equivalently by discarding the DC bin
of their DFT. Collect the two axes into complex variables, and define 
\[
\begin{aligned}
\bv(t) &:= \Delta v_d(t) + j \Delta v_q(t), \\
\bi(t) &:= \Delta i_d(t) + j \Delta i_q(t).
\end{aligned}
\]
Straightforward algebraic manipulations \cite{Martin2004} then show that
\begingroup\makeatletter\def\f@size{9.5}\check@mathfonts
\begin{equation}\label{eq:complex_model}
\bv(t) = \bG_+(\p) \bi(t) + \bG_-(\p) \bi^\ast(t),
\end{equation}
where $\bi^\ast$ denotes the complex conjugate of $\bi$ and
\begin{equation}\label{eq:complex_tfs}
\begin{aligned}
   \bG_+(\p) &=  0.5 \big({Z_{dd}(\p) + Z_{qq}(\p)} + j ({Z_{qd}(\p)  - Z_{dq}(\p)}) \big),   \\ 
   \bG_-(\p) &= 0.5 \big({Z_{dd}(\p) - Z_{qq}(\p)} + j({Z_{dq}(\p)  + Z_{qd}(\p)})\big),
\end{aligned}
\end{equation}
\endgroup
are two SISO complex transfer operators.  The subscripts indicate the direction of rotation: a $dq$ perturbation $\bi(t) = \eu^{j\omega t}$ rotates positively, and its conjugate negatively. The TF representations follow by Laplace transformation,
under which $\p$ becomes $s$, and the frequency response is recovered at
$s = j\omega$, where $j$ is the imaginary unit. Note that
\eqref{eq:complex_model} is not $\mathbb{C}$-linear:
the conjugate $\bi^\ast$ enters through its own channel, and it is precisely
this conjugate coupling that causes the $dq$-asymmetry. 

When $Z_g(\p)$ is symmetric, i.e.,  $Z_{dd}(\p) = Z_{qq}(\p) = G_d(\p)$ and
$Z_{qd}(\p) = -Z_{dq}(\p) = G_q(\p)$, the second channel vanishes,
$\bG_-(\p) = 0$, and \eqref{eq:complex_model} reduces to
$\bv(t) = \bG(\p) \bi(t)$~with
\begin{equation}
\label{eq:single_complex_tf}
\bG(\p) = G_d(\p) + j G_q(\p) = \bG_+(\p).
\end{equation}
Accordingly, hereafter, $G_+(p)$ is referred to as the direct TF term, as it represents the direct $\bi(t)$-to-$\bv(t)$ channel, whereas $G_-(p)$ is referred to as the coupling TF term, as it multiplies the conjugate $\bi^*(t)$ and vanishes in the symmetric case.

A single SISO complex TF then represents $Z_g$ completely, and one set of
periodic measurements suffices to estimate it non-parametrically. In the asymmetric case, $\bi^\ast$ is always present, producing a model with
twice as many unknowns as data equations, which cannot be resolved from one
measurement set unless something further is assumed about the frequency
response. That asymmetry arises naturally in a grid containing actively controlled
devices. It is generated by the control loops themselves, synchronization and
outer loops acting differently on the two mirror frequencies. This is the same
mechanism that motivates the modified sequence-domain and unified-impedance
representations \cite{rygg2016modified,Harnefors2007,Wang2018,Harnefors2020}.

\subsection{Grid-connected converter}

A grid-connected voltage-sourced converter with an $LCL$ filter serves as the
excitation source throughout, as drawn in Fig.~\ref{figure:fig1}; the DC link is
treated as ideal. Current control is carried out in the $dq$ frame,
synchronization is provided by a PLL locked to the filter capacitor voltage, providing phase angle $\theta_{\mathrm{pll}}$, and
the current-controller output $(i_\mathrm{c,D}, i_\mathrm{c,Q})$ becomes the converter voltage reference $(u_\mathrm{c,D}^\mathrm{ref}, u_\mathrm{c,Q}^\mathrm{ref})$  realized
through PWM. Outer power loops are also present 
 but omitted from the figure.

Excitation is introduced by superimposing a wideband signal on the converter
voltage references which allows for a higher excitation bandwidth than if added to the current reference $(i_\mathrm{c,D}^\mathrm{ref}, i_\mathrm{c,Q}^\mathrm{ref})$. The signal used here is a zero-mean random binary sequence (RBS) \cite{ljung1998system}. Unlike an MLBS, it is aperiodic and may be truncated at any length. The proposed approach does not depend on a specific excitation signal. Any signal which is sufficiently exciting in the
sense of Section~\ref{sec:excitation} can be used, that is, provided the phase of the
resulting spectrum varies widely over frequency. This leaves room
to shape the excitation, for instance to place more of the injected energy where
the response is most needed, at a given amplitude and record length. The design of the excitation spectrum is
outside the scope of this work.

\subsection{Measurement setup}
\begin{figure}[t]
    \centering
\resizebox {0.8\linewidth} {!} {
	\begin{scaletikzpicturetowidth}{0.9\linewidth}
	\input{figures/CL-block-diagram}%
	\end{scaletikzpicturetowidth}
	}
    \caption{Measurement setup. Voltage and current noise represent errors due to inaccuracies of the measurement devices $S_i$ and $S_v$. The converter's switching harmonics act as an additional excitation signal. The switches represent signal samplers, and the voltage and current noise represent measurement noises.}
    \label{figure:measurment_setup}
\end{figure}
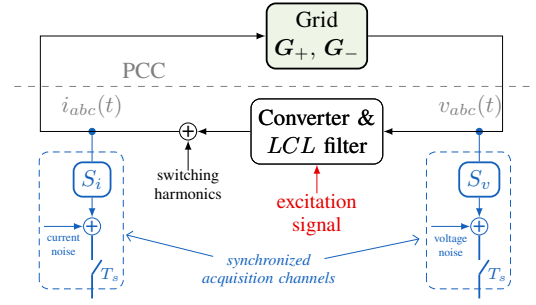

Fig.~\ref{figure:measurment_setup} shows the assumed measurement configuration. Recording of
$v_{abc}(t)$ and $i_{abc}(t)$ begins when the excitation is applied and runs for
its duration. Anti-alias filtering ahead of
the sampler is necessary, and the two channels must be synchronized and
relatively calibrated so that the transducer responses $S_i$ and $S_v$ can be accounted for.
Sampling (indicated using switches in Fig.~\ref{figure:measurment_setup}) is uniform and fast enough for the band of interest, for which the
converter's own period $T_s$ (typically the switching period or half of it
\cite{Buso2022}) is a natural choice. Park's transformation then maps the
samples into the synchronous frame. Section~\ref{sec:measurement} analyzes and mitigates two nuisances of this signal processing pipeline: placing the anti-alias filter \emph{before} the Park
transform, and fixing the frame in which the estimate is expressed.

\subsection{Choice of reference frame}\label{sec:frames}

Park's transformation requires the phase angle of the grid fundamental voltage, which is usually obtained from a PLL, and the dynamics of the PLL can distort the measurement if they are not explicitly accounted for. To bypass this, some works identify the impedance in the modified sequence domain or in the stationary frame with an augmented model. At a constant grid frequency, however, the $dq$ frame impedance and these two alternatives are connected by linear transformations
\cite{rygg2016modified,Harnefors2007,Wang2018}, carry identical information, require the same fundamental phase information and
differ only in how it enters the model: as the frame alignment itself in $dq$; as an explicit $\eu^{j2\theta}$ factor in the augmented stationary model; as a DFT window synchronized to the fundamental in the modified sequence domain. Identification in the $dq$ frame with a static phase alignment reduces the dependence to a single angle, obtained as described in
Section~\ref{sec:alignment}.

%% file: figures/inverter_scheme_small.tex
\resizebox {\columnwidth} {!} {

\begin{tikzpicture}[circuit ee IEC,scale=0.45, every node/.style={scale=0.57}]
\draw (2.5,3.6) -- (3.5,3.6);
\draw (2.5,1.6) -- (3.5,1.6);
\draw  (3.5,3.8) node (v13) {} rectangle (5.3,1.4);
\draw (4.6,3.5) -- (4.6,3.1) -- (4.1,2.85) node (v3) {};
\draw[-latex](4.1,2.35) -- (4.6,2.1) ;
\draw (4.6,2.1) -- (4.6,1.7);
\draw (4.1,3.1) -- (4.1,2.1);
\draw (4,3.1)--(4,2.1);
\draw (3.7,2.6) -- (4,2.6);
\draw (6,2.6) node (v7) {} to [inductor={yscale=1.5}] (8,2.6) node (v8) {};
\draw (8.9,2.6) node (v7) {} to [inductor={yscale=1.5}] (10.9,2.6) node (v8) {};
\draw (5.3,2.6) -- (6.1,2.6) node (v1) {}; 
\draw (10.6,2.6) -- (13,2.6) -- (14.4,2.6);
\node at (2.7,3.3) {$+$};
\node at (2.7,1.8) {$-$};
\node at (2.8,2.6) {$v_\mathrm{dc}$};
\node at (7,3.2) {$L_\mathrm{f,1}$};
\draw(2.45,3.6) circle (0.6 mm); 
\draw(2.45,1.6) circle (0.6 mm); 
\draw[-latex] (4.4,0.6) -- (4.4,1.4);
\draw[rounded corners=3,fill=black!7,draw=black!55]  (3.5,0.6) rectangle (5.3,-0.2);
\node at (4.4,0.2) {PWM};
\node at (4.4,4.15) {VSC};
\draw [-latex](4.4,-1.4) -- (4.4,-0.2);
\draw  (3.7,-1.4) rectangle (5.1,-2.3);
\draw (3.7,-2.3) -- (5.1,-1.4);
\node [scale= 0.9] at (4.1,-1.65) {abc};
\node [scale = 0.7] at (4.7,-2.05) {DQ};
\draw[-latex] (4,-4) -- (4,-2.3); 
\draw[-latex] (4.8,-3.3) -- (4.8,-2.3);
\draw (4.8,-3.45) circle(1.5mm);
\draw (4,-4.15) circle(1.5mm);
\draw[-latex] (6.2,-3.45) -- (4.95,-3.45);
\draw [-latex] (6.2,-4.15)--(4.15,-4.15) ;
\draw[-latex,red!90!black] (4,-4.9) -- (4,-4.3); 
\draw [-latex,red!90!black](4.8,-4.9) -- (4.8,-3.6);
\node [red!90!black] at (4.4,-5.2) {\emph{excitation signal}};
\draw [rounded corners = 3,fill=black!7,draw=black!55] (6.2,-3.1) rectangle (9.9,-4.5);
\draw (8.4,2.6) -- (8.4,1.95);
\fill (8.4,2.6) circle(0.6mm);
\draw (8.1,1.95) -- (8.7,1.95); 
\draw (8.1,1.8) -- (8.7,1.8); 
\draw (8.4,1.8) -- (8.4,1.2);
\draw (8.25,1.2) -- (8.55,1.2);
\draw (8.35,1.15) -- (8.45,1.15);
\node at (7.7,1.85) {$C_\mathrm{f}$};

\draw[color=black]  (5.6,2.6) ellipse (0.05 and 0.15);

\draw  (6.2,-1.4) rectangle (7.6,-2.3);
\draw[-latex,color=black] (5.6,2.45) --(5.6,-1.9)-- (6.2,-1.9);
\draw (6.2,-2.3) -- (7.6,-1.4);
\draw [-latex] (3,-1.85) -- (3.7,-1.85);
\node at (2.55,-1.85) {$\theta_\mathrm{pll}$};
\node [scale=0.9] at (6.6,-1.65) {abc};
\node [scale=0.7] at (7.2,-2.05) {DQ};

\node[color=black] at (6.25,-0.7) {$i_\mathrm{c,abc}$};
\draw [-latex,color=black] (8.9,2.6) -- (8.9,-1.4);

\fill [color=black](8.9,2.6) circle(0.6mm);
\draw [rounded corners = 3,fill=black!7,draw=black!55] (8.5,-1.4) rectangle (9.9,-2.3);
\node at (9.2,-1.85) {PLL};
\draw[-latex] (6.4,-2.3) -- (6.4,-3.1); 
\draw [-latex](7.4,-2.3) -- (7.4,-3.1);
\node at (6.9,-2.75) {$i_\mathrm{c,D}$};
\node at (7.9,-2.75) {$i_\mathrm{c,Q}$};
\draw [-latex](8.5,-1.85) -- (7.6,-1.85);
\node at (8.1,-1.5) {$\theta_\mathrm{pll}$};
\node at (8.1,-3.55) {$DQ$-current};
\node [color=black] at (8.2,-0.7) {$v_\mathrm{c,abc}$};
\node at (8.1,-4.05) {control};
\draw[-latex] (10.6,-3.45) -- (9.9,-3.45);
\draw [-latex](10.6,-4.15) -- (9.9,-4.15);
\node at (11.1,-3.45) {$i_\mathrm{c,D}^\mathrm{ref}$};
\node at (3.5,-2.75) {$u_\mathrm{c,Q}^\mathrm{ref}$};
\node at (5.35,-2.75) {$u_\mathrm{c,D}^\mathrm{ref}$};
\node at (11.1,-4.15) {$i_\mathrm{c,Q}^\mathrm{ref}$};

\draw [color=backgroundcolor] (12.5,2.6) ellipse (0.05 and 0.15);
\draw[-latex,color=backgroundcolor] (12.5,2.45) -- (12.5,2) -- (12.5,0.4);
\draw [-latex,color=backgroundcolor](11.5,2.6) -- (11.5,0.4);
\fill[color=backgroundcolor] (11.5,2.6) circle(0.6mm);
\node [color=backgroundcolor] at (12,-1.65) {measured};
\node  [color=backgroundcolor] at (12,-2.05) {data};
\draw[fill=greenBackgroundcolor!30,color=greenBackgroundcolor!30,opacity=0.8] (13.5,1.6) node (v2) {} -- (13.5,3.8) -- (15.3,3.8) -- (15.3,4.1) -- (15.8,3.65) -- (15.3,3.2) -- (15.3,3.5) -- (13.8,3.5) -- (13.8,1.6) -- (13.5,1.6);
\node [color=greenBackgroundcolor] at (14.55,4.15) {$Z_\mathrm{g}(s)$};

\node[backgroundcolor] at (12.05,2.2) {$i_\mathrm{abc}$};

\draw  (14.4,2.8) rectangle (15.5,2.4);
\draw (15.5,2.6) -- (16.2,2.6) -- (16.2,2.2);
\draw  (16.2,1.9) ellipse (0.3 and 0.3);
\draw (16.2,1.6) -- (16.2,1.2);
\draw (16.05,1.2) -- (16.35,1.2);
\draw (16.15,1.15) -- (16.25,1.15);
\draw  [densely dashed, rounded corners =5,color=greenBackgroundcolor, thick](14.1,3) rectangle (16.7,0.9);
\node [color=greenBackgroundcolor] at (15.4,0.55) {three-phase};
\node [color=greenBackgroundcolor] at (15.4,0.1) {grid};
\draw  plot[smooth, tension=.7] coordinates {(16,1.85) (16.1,2.05) (16.3,1.75) (16.4,1.95)};
\draw (13,2.8) -- (13,2.4);
\node at (12.85,3.15) {PCC};

\draw (4.6,1.9) -- (4.9,1.9) -- (4.9,2.5); 
\draw (4.9,2.8) -- (4.9,3.3) -- (4.6,3.3); 
\draw (4.75,2.8) -- (5.05,2.8); 
\draw (4.9,2.8) -- (4.7,2.5) -- (5.1,2.5) -- (4.9,2.8);

\draw (7.8,2.6) -- (8.9,2.6);
\node at (9.9,3.2) {$L_\mathrm{f,2}$};

\node [backgroundcolor] at (11,2.2) {$v_\mathrm{abc}$};
\draw  [color=backgroundcolor](11.3,0.4) rectangle (12.7,-0.5);
\draw [color=backgroundcolor](11.3,-0.5) -- (12.7,0.4);
\node  [scale=0.9,backgroundcolor] at (11.7,0.15) {abc};
\node  [scale=0.9, backgroundcolor] at (12.3,-0.2) {dq};
\draw [-latex,backgroundcolor](11.5,-0.5) -- (11.5,-1.3); 
\draw [-latex,backgroundcolor](12.5,-0.5) -- (12.5,-1.3);
\node [backgroundcolor]at (11,-0.8) {$v_\mathrm{dq}$};
\node [backgroundcolor]at (12.1,-0.8) {$i_\mathrm{dq}$};
\draw [-latex,backgroundcolor](10.6,-0.05) -- (11.3,-0.05);
\node [backgroundcolor]at (10.3,-0.05) {$\theta$};
\end{tikzpicture}

}

%% file: figures/CL-block-diagram.tex
\begin{tikzpicture}[ auto, >=latex, scale = \tikzscale,
block/.style={draw,rectangle,minimum height=1cm,minimum width=1.0cm, rounded corners},
link/.style = {->},
knot/.style={draw,circle,inner sep=0pt,minimum size=0.1mm,fill=black,draw=black},
sum/.style={draw,circle,inner sep=0mm,minimum size=0.02}, 
small_block/.style={draw,rectangle,thick,minimum height=1em,minimum width=1em, rounded corners},
]
\node[block,align=center] (converter)              {Converter \& \\$LCL$ filter};
\node[block,align=center] (converter)              {Converter \& \\$LCL$ filter};
\node[sum]                (converter_harmonics) at ($(converter.west) - (1.3,0)$) {+};
\node[block,align=center,fill= greenBackgroundcolor!12, thick] (grid)                at ($(converter.north) + (0,1.3)$) { Grid\\ $\bG_+, \,\bG_-$};
\node[small_block,align=center, color = backgroundcolor] (sensor_Gv)      at ($(converter.east) + (2,-1)$) { $S_v$};
\node[small_block,align=center, color = backgroundcolor] (sensor_Gi)      at ($(converter.west) - (3.3,1)$) { $S_i$};
\node[sum, color = backgroundcolor]                (voltage_noise)       at ($(sensor_Gv.south) - (0,0.6)$) {+};
\node[sum, color = backgroundcolor]                (current_noise)       at ($(sensor_Gi.south) - (0,0.6)$) {+};
\node[knot, color = backgroundcolor]               (sensor_i) at ($(converter.west) - (3.3,0)$) {.};
\node[knot, color = backgroundcolor]               (sensor_v) at ($(converter.east) + (2,0)$) {.};

\draw[link, color = red!90!black] ($(converter.south) - (0,0.6)$) -- node [below, pos = 0.25, font = \small, align = center] {excitation\\ signal} (converter.south);
\draw[link] (converter.west) -- (converter_harmonics.east);
\draw[link] (converter_harmonics.west) -- ($ (converter_harmonics.west) - (2.89,0) $) |- (grid.west);
\draw[link] (grid.east) -- ($ (grid.east) + (2.8,0) $) |- (converter.east);
\draw[link] ($(converter_harmonics.north) - (0,1)$) -- node [pos = 0.3, below, align=center, font = \scriptsize] {switching\\harmonics}  (converter_harmonics.south) ;
\draw[link, color = backgroundcolor] (sensor_v) -- (sensor_Gv) -- (voltage_noise);
\draw[link, color = backgroundcolor] (sensor_i) -- (sensor_Gi) -- (current_noise);

\draw[link, color = backgroundcolor] ($(voltage_noise) - (1,0)$) -- node [midway,below, align=center, font = \tiny] {voltage \\ noise}(voltage_noise);
\draw[link, color = backgroundcolor] ($(current_noise) - (1,0)$) -- node [midway,below, align=center, font = \tiny] {current \\ noise}(current_noise);

\draw[thick, color = backgroundcolor] (current_noise) -- ($(current_noise) - (0,0.7)$);
\draw[thick, color = backgroundcolor] ($(current_noise) - (-0.2,0.7)$) -- node [at end, right, font = \scriptsize] {$T_s$} ($(current_noise) - (0,1)$);
\draw[thick, color = backgroundcolor] (voltage_noise) --  ($(voltage_noise) - (0,0.7)$);
\draw[thick, color = backgroundcolor] ($(voltage_noise) - (-0.2,0.7)$) --  node [at end, right, font = \scriptsize] {$T_s$} ($(voltage_noise) - (0,1)$);
\draw[thick, color = backgroundcolor] ($(current_noise) - (0,1)$) -- ($(current_noise) - (0,1.5)$);
\draw[thick, color = backgroundcolor] ($(voltage_noise) - (0,1)$) -- ($(voltage_noise) - (0,1.5)$);

\draw[rounded corners=4pt, densely dashed, gray, color = backgroundcolor]  ($(voltage_noise) - (1.1,1.3)$) rectangle++(1.75,2.85);
\draw[rounded corners=4pt, densely dashed, gray, color = backgroundcolor] ($(current_noise) - (1.1,1.3)$) rectangle ++(1.75,2.85);
\draw[dashed, thin, gray] ($(converter.north) + (-6.3,0.25)$) -- node [ pos = 0.15, below] {$i_{abc}(t)$}  node [ pos = 0.88, below] {$v_{abc}(t)$}  node [ pos = 0.25] {PCC} ($(converter.north) + (4.5,0.25)$);
\node[ font = \scriptsize, align = center, backgroundcolor] (text) at ($(converter.south) - (1,2.2)$)  {\emph{synchronized} \\ \emph{acquisition channels}};
 \draw[->, shorten <= -3mm, shorten >= 4mm, >= Straight Barb, backgroundcolor] (text) -- (current_noise);
 \draw[->, shorten <= -3mm, shorten >= 8mm,>=Straight Barb, backgroundcolor] (text) -- (voltage_noise);
\end{tikzpicture}

%% file: sections/identification.tex
\section{Impedance Identification Method}\label{sec:analysis}

Given a data set of a single arbitrary record,
\begin{equation}\label{eq:dataset}
D_N :=\{ (\bv(t_n), \bi(t_n)),  n \in \{0, \dots, N-1\} \},
\end{equation}
the goal is to construct a \emph{non-parametric} estimator of
$\bG_+(j\omega), \bG_-(j\omega)$ on the uniform grid
$\omega_k\! =\! \frac{2\pi k}{NT_s}$, viz.
\[
D_N \, \mapsto \, \{ (\, \widehat{\bG_+(j\omega_k)},\; \widehat{\bG_-(j\omega_k)}\,), \;\, k \!\in\!\{ 0, \dots, N-1\} \}.
\]
To this end, let $\bV\!_k$ and $\bI_k$ denote the $N$-point DFT of
$\{ \bv(t_n) \}$ and $\{ \bi(t_n) \}$,
\begin{equation}\label{eq:dft_spectrum}
\bV\!_k = \frac{1}{\sqrt{N}} \sum_{n=0}^{N-1} \bv(t_n) \eu^{-j\omega_k n T_s}, \quad k \in \{0, \dots, N-1\},
\end{equation}
and similarly for $\bI_k$. Recall that the $N$-point DFT can be computed very
efficiently using fast Fourier transform algorithms, and that the spectra are
periodic with period $N$, namely $\bI_{k+N} = \bI_{k}$. Because the time-domain
signals here are complex-valued, the Hermitian symmetry property of the DFT does
not hold, but it still satisfies conjugacy and reversal properties, that is, the
DFT of $\{\bi^\ast(t_n)\}$ is given by $\{\bI^\ast_{(N-k)_N}\}$ where $(N-k)_N$
stands for ``$N-k$ modulo $N$''. This index is equal to $k'$ in the interval $0$
to $N-1$ that can be written as $k = k' + \nu N$ where $\nu$ is an integer;  see Appendix~\ref{app:conj}.

\subsection{Relation between voltage and current DFT spectra}

The starting point is the following exact relation, which requires neither
periodicity of the data nor any approximation.

\begin{theorem}\label{thm:dft}
Consider the model (1). Let $\bG_+$ and $\bG_-$ be stable and causal, and let $\bv$ and $\bi$ be
observed over the finite time interval $[0, NT_s]$, so that $\bV\!_k$ and
$\bI_k$ in \eqref{eq:dft_spectrum} are the DFT spectra of that record. Then, for
every $k \in \{0,\dots,N-1\}$,
\begin{equation}\label{eq:dft_relation}
\bV\!_k = \bG_+(j\omega_k) \bI_k + \bG_-(j\omega_k) \bI_{(N-k)_N}^\ast + \bT(j\omega_k),
\end{equation} 
where the transient term $\bT(j\omega_k)$ accounts for the spectral leakage
which arises from the mismatch in initial and final conditions, and for the
aliasing effects due to the truncation to $[0, NT_s]$. It collects the
unforced responses of $\bG_+, \bG_-$.  
\end{theorem}

The explicit derivation of (6) is given in Appendix~\ref{app:dft}. The transient term decays at a
rate $\mathcal{O}(N^{-\frac{1}{2}})$, similar to the real-signal case, as shown in \cite[Thm.~2.1]{ljung1998system} and
\cite{Pintelon1997}. Here, we establish the extension of the relation to complex, conjugate-coupled signals, and in particular the reversed and
conjugated spectrum $\bI_{(N-k)_N}^\ast$ in the second channel.

Notice that the TFs $\bG_+$ and $\bG_-$  are built as linear combinations of the same four real TFs in \eqref{eq:complex_tfs}, so their poles are among those of $Z_g$. The two sets need not coincide; indeed for a $dq$-symmetric system, $\bG_- = 0$. 

Relation \eqref{eq:dft_relation} holds at each spectral line on its own, and the
method proposed in the sequel is built on it. It returns an estimate of $\bG_+(j\omega_k)$ and $\bG_-(j\omega_k)$ at every spectral line $\omega_k$ separately, and the estimates at one line are
tied to the estimates at another only through the record $D_N$ from which both are
computed. Nothing is assumed about the number or the locations of the poles
 of $Z_g$, its relative degree, or the topology of the network
behind the PCC: no global TF of fixed order is assumed or fitted to the whole band $[0, \omega_{N-1}]$. The method delivers a list of complex numbers corresponding to the estimates on the grid $\omega_k$, of the same kind as those delivered by a frequency sweep. This is known as \emph{non-parametric} estimation in the system identification literature \cite{ljung1998system}.

Identifying the equivalent impedance from a single \textit{short} measurement
cycle using \eqref{eq:dft_relation} presents two challenges. The first stems
from the error introduced by the transient term $\bT$, which does not vanish for
a short record; applying a window suppresses it only in part, and may distort the
resonances that are to be resolved. The second is that \eqref{eq:dft_relation}
provides $N$ complex data equations, while the responses to be determined amount
to $2N$ complex unknowns, so a single record leaves the problem
underdetermined. 

Classical impedance measurement addresses both issues, but at the cost of an extended
measurement time. The first is usually addressed by waiting for the transient to decay
and recording an integer number of periods. The second is addressed by a second
measurement with an independent perturbation
\cite{francis2011algorithm}, which can be a single tone \cite{huang2009small} or
a wideband \cite{shen2013three,Rygg2016}. However, these solutions lengthen the measurement time and require the grid to be unchanged throughout.

The method proposed here addresses the two challenges without lengthening the
record. The first is addressed by estimating $\bT$ together with the complex TFs instead of suppressing it or waiting for it to decay. The second is addressed by local parametric modeling, which provides the missing equations by using neighboring spectral lines.

\subsection{Local parametric modeling} 

To obtain a non-parametric estimate at a frequency $\omega_k$, the complex
TFs $\bG_+(j\omega_k), \bG_-(j\omega_k)$ and $\bT(j\omega_k)$ are approximated over a short
frequency range around $\omega_k$ using a parametric model of low order.
Every line in that range then contributes an equation to the estimate at
$\omega_k$. Estimates at different frequencies are correlated only via raw data, and
therefore the method remains truly non-parametric in nature. This is the idea of
local parametric modeling \cite{Pintelon2010,McKelvey2012,Pintelon2021}.

By Theorem~\ref{thm:dft}, the poles of  $\bG_+$, $\bG_-$ and $\bT$ are confined to those of $Z_g(s)$.
Therefore, for a local frequency interval $[\omega_{k-\ell}, \omega_{k+\ell}]$ centered on $\omega_k$ with a chosen radius $\ell$, we approximate
\begingroup\makeatletter\def\f@size{9.5}\check@mathfonts
\begin{equation}\label{eq:LPM}
\begin{aligned}
        \bG_+(j\omega_{k+r}) &\approx \frac{\bB_k^+(r)}{\bA_k(r)}, \quad \bG_-(j\omega_{k+r}) \approx \frac{\bB_k^-(r)}{\bA_k(r)}\\
          \bT(j\omega_{k+r}) &\approx \frac{\bC_k(r)}{\bA_k(r)},
\end{aligned}
\end{equation}
\endgroup
where $r \in \{-\ell, \dots, \ell\}$ and $\bA_k, \bB^+_k, \bB^-_k, \bC_k$ are
complex polynomials in $r$ of degree $R$, parameterized as
\begin{equation}\label{eq:Apoly}
\bA_k(r) := \sum_{s=0}^R a_s(k)\, r^{s}, \qquad a_s(k) \in \mathbb{C},
\end{equation}
and similarly for the other polynomials with coefficients $b^+_s(k), b^-_s(k)$, and $c_s(k)$, respectively. 
Using a single fixed degree for all polynomials across $k$ simplifies the approach while effectively producing a local rational model of order $R$.

The use of a common denominator  $\bA_k$ is justified by Theorem~\ref{thm:dft}, as it needs only to account for the poles of $Z_g$ lying near the interval. The implicit assumption here is that a rational model of degree $R$ can closely approximate the three functions over the interval $[\omega_{k-\ell}, \omega_{k+\ell}]$.

Estimates at $\omega_k$ are obtained by evaluating \eqref{eq:LPM} at the center
of the interval, $r = 0$. Because the three ratios in \eqref{eq:LPM} are
unchanged when $\bA_k$, $\bB^+_k$, $\bB^-_k$ and $\bC_k$ are scaled by a common
complex factor, we may set $a_0(k) = 1$ for every $k$, which removes that
freedom and makes the denominator unity at $r = 0$. With this parameterization, the vector of parameters to be estimated is
\begin{align*}
&\theta_k :=\left[\begin{matrix}
    a_1(k) \cdots a_R(k) & b_0^+(k)    \cdots\end{matrix}\right. \\
& \hspace{2cm}  \left.\begin{matrix} \cdots b_R^+(k) & b_0^-(k)\cdots  b_R^-(k) & c_0(k) \cdots c_R(k) \end{matrix}\right].
\end{align*}
It contains $4R + 3$ unknown complex parameters. From \eqref{eq:dft_relation}
and \eqref{eq:LPM}, the local model linking the DFT spectra is
\begin{equation}\label{eq:nonlinear}
\begin{aligned}
    \bV\!_{k+r} = \,&
    \frac{\bB_k^+(r)}{\bA_k(r)} \bI_{k+r}
    + \frac{\bB_k^-(r)}{\bA_k(r)} \bI^\ast_{(N-k-r)_N}\\
    & + \frac{\bC_k(r)}{\bA_k(r)} + \widetilde{\bE}_{k+r},
\end{aligned}
\end{equation}
where $\widetilde{\bE}_{k+r}$ collects the local parametric interpolation error
of the three approximations in \eqref{eq:LPM}. This relation is nonlinear in the
parameters. Multiplying it by $\bA_k(r)$ gives
\begin{equation}\label{eq:localmodel}
\begin{aligned}
 \!\!\!\!\!   \bA_k(r) \bV\!_{k+r} = \,&
    \bB_k^+(r) \bI_{k+r}
    + \bB_k^-(r) \bI^\ast_{(N-k-r)_N}\\& + \bC_k(r) + \bE_{k+r},
\end{aligned}
\end{equation}
which is linear in the
parameters. The parameters for each $k$ can then be estimated by minimizing
$\bE_{k+r}$ in the least squares sense; this corresponds to fitting a
\textit{local} ARX model to the spectra only over
$[\omega_{k-\ell}, \omega_{k+\ell}]$.

The distinction from a global parametric model is worth making explicit, since a
parametric model is fitted at every line. A global model is a model valid for \textit{all} frequencies: one TF of a chosen order is fitted to the whole band of interest, whose
order fixes how many poles can be represented anywhere in the band. A poorly chosen order results in an error that propagates to every frequency. The model in
\eqref{eq:LPM} is a model of the frequency response over
$[\omega_{k-\ell}, \omega_{k+\ell}]$, and is used to estimate the response only at the central frequency $\omega_k$. It is discarded once the parameters at
$\omega_k$ have been estimated, and a separate (local) model is fitted at the next line. The parameter
$R$ therefore sets how much curvature the three functions may exhibit across one
short interval, rather than how many poles $Z_g$ is allowed to have.
Section~\ref{sec:order} examines the sensitivity of the estimates to this choice.

Define the complex column vectors
\[
\begin{aligned}
\bY\!_k &:= \begin{bmatrix}
\bV\!_{k-\ell}\!\! &\dots &\!\!\bV\!_{k+\ell}
\end{bmatrix}^\top\!\!, \quad
\bU^+_k := \begin{bmatrix}
\bI_{k-\ell} \!\!&\dots\!\! &\bI_{k+\ell}
\end{bmatrix}^\top\!\!,\\
\bU^-_k &:= \begin{bmatrix}
\bI^\ast_{(N-k+\ell)_N} &\dots &\bI^\ast_{(N-k-\ell)_N}
\end{bmatrix}^\top,\\
\bEps\!_k &:= \begin{bmatrix}
\bE_{k-\ell}\!\! &\dots &\!\!\bE_{k+\ell}
\end{bmatrix}^\top\!\!,
\end{aligned}
\]
and the Vandermonde-type matrices
\[
\widetilde{\Phi} := \begin{bmatrix}\tilde\phi(-\ell) & \cdots & \tilde\phi(\ell)\end{bmatrix}^\top \!\!\in \real^{(2\ell+1)\times R},
\qquad
\Phi := \begin{bmatrix}\1 & \widetilde{\Phi}\end{bmatrix},
\]
with $\tilde{\phi}(r) := \begin{bmatrix} r & r^2&  \dots& r^R\end{bmatrix}^\top$.
The data matrix $\Phi_k$ relating the vector of parameters $\theta_k$ and voltage DFT spectra $\bY_k$, as 
$\bY_k = \Phi_k \theta_k + \bEps_k$, is  given by
\begin{equation}\label{eq:datamatrix}
\begin{aligned}
&\Phi_k := \left[\begin{matrix}
    -\diag(\bY\!_k)\,\widetilde{\Phi}  &\diag(\bU^+_k)\,\Phi  \;  \cdots\end{matrix}\right. \\
& \hspace{4cm}  \left.\begin{matrix} \cdots \;     \diag(\bU^-_k)\,\Phi &    \Phi \end{matrix}\right],
\end{aligned} 
\end{equation}
a complex $(2\ell+1)\times(4R+3)$ matrix whose four blocks correspond to the
coefficients of $\bA_k, \bB^+_k, \bB^-_k$ and $\bC_k$ in $\theta_k$,
respectively. The parameter estimate is then obtained by solving the linear
least-squares problem $\min_\theta\| \bY\!_k - \Phi_k \theta  \|$, with the 
closed-form solution
\begin{equation}\label{eq:LS_sol}
\hat{\theta}_k = \Phi_k^\dag \bY\!_k.
\end{equation}
The matrix $\Phi_k^\dag$ denotes the pseudo-inverse of $\Phi_k$, which may be computed using a singular value decomposition with an appropriate scaling to improve numerical conditioning. Setting $r = 0$ in \eqref{eq:LPM}, every polynomial reduces to its constant coefficient, so that
$\bA_k = 1$ by normalization and $\bB^\pm_k = b^\pm_0(k)$, and the two ratios collapse to their numerator constants. This gives the estimators
\begin{equation}\label{eq:estimators}
 \widehat{\bG_+(j\omega_{k})} = \widehat{b_0^+(k)}, \qquad \widehat{\bG_-(j\omega_{k})} = \widehat{b_0^-(k)},
\end{equation}
where $\widehat{b_0^+(k)}$ and $\widehat{b_0^-(k)}$ are extracted from
$\hat{\theta}_k$. The remaining $4R+1$ parameters are discarded and not reported.

\begin{remark} The method estimates the response at each frequency $\omega_k$ separately. The least-squares problems are formed separately over
$k$, and the computations may be parallelized.
\end{remark}

\subsection{Identifiability and excitation}\label{sec:excitation}

The solution \eqref{eq:LS_sol} is unique if and only if $\Phi_k$ has full column
rank. This imposes two conditions: one on the width of the local frequency
interval, and one on the local spectra $\bY\!_k$, $\bU^+_k$ and $\bU^-_k$.

The first is a counting condition. There are $4R+3$ complex parameters and
$2\ell+1$ complex data equations, so that necessarily
\begin{equation}\label{eq:count}
2\ell+1 \geq 4R+3 .
\end{equation}

The second condition concerns the data, and is the frequency-domain counterpart
of the classical requirement for identification in the presence of unknown
initial conditions \cite{Mathew1972}. If $\Phi_k$ is rank deficient, the local
spectra satisfy a polynomial relation of degree at most $R$ over the interval;
conversely, rank is lost whenever the local current spectrum and its mirror
satisfy
\begin{equation}\label{eq:ratform}
\begin{aligned}
\beta^+(r) \bI_{k+r} +& \beta^-(r) \bI^\ast_{(N-k-r)_N} + \gamma(r) = 0, \;\\
 r &\in \{-\ell,\dots,\ell\},
\end{aligned}
\end{equation}
for complex polynomials $\beta^+, \beta^-, \gamma$ of degree at most $R$, not all
zero, in which case the columns of $\Phi_k$ corresponding to the input are not
distinguishable from those corresponding to $\bT$. The measurement is sufficiently
exciting at $\omega_k$ in the specific sense that its spectrum admits no such
description, and Appendix~\ref{app:rank} gives the precise statement and its
proof.

\smallskip

\subsubsection*{The RBS excitation}
The spectrum of an RBS excitation has  a phase that varies irregularly from one spectral line to the
next (it is random in $[-\pi, \pi]$).  The spectrum of the current $\bi$ at the PCC is fixed by the response of the closed loop to the RBS
realization. A polynomial of degree $R$ is fixed by $R+1$
values, so a relation of the form \eqref{eq:ratform} would require the remaining
$2\ell-R$ lines of the interval to satisfy the same polynomial, which does
not occur for a spectrum of this kind due to the complexity of the closed loop system. The counting condition \eqref{eq:count} leads to the same conclusion: a wider interval brings in more non-zero spectral lines, and
each of them is one more constraint that a degenerate relation would have to
satisfy.

\smallskip
\subsubsection*{The estimate at $\omega_g$} One spectral line requires separate treatment. In the $dq$ frame this is the
line at $\omega = 0$, which corresponds to the fundamental $\omega_g$ in $abc$
coordinates, and which is not excited in the data.  At $\omega = 0$, the DFT line is its own mirror, $(N-k)_N = k$, so the direct and conjugate regressors there are $\bI_0$  and $\bI_0^*$. Because the mean has been removed from the records, both $\bV_0$ and $\bI_0$ vanish, so the corresponding row of \eqref{eq:datamatrix} reduces to $c_0 = 0$: it carries no information about either channel, and it forces the constant coefficient of $\bC_k$ to zero, thereby biasing the transient term. Appendix~\ref{app:dcbin} works this out on a simple example. 

The corresponding row in $\Phi_k$ is therefore removed from every local interval in which it appears. This
is possible only when the interval is wider than the minimum width that the  counting condition
\eqref{eq:count} requires: at the minimum width $\ell = 2R+1$ the interval
carries exactly $4R+3$ lines, and removing one would leave fewer equations than
parameters. On the other hand, with $\ell = 4R+2$, twice the minimum, it carries $8R+5$ lines
against $4R+3$ parameters, so the loss of one line is immaterial. Notably, the
equivalent impedance is still estimated at $\omega_g$, because the local model
is fitted over an interval on which the neighboring lines are excited.

\subsection{Special cases}

Three special cases are of interest. If the excitation signal is repeated
periodically and $D_N$ contains an integer number of steady-state periods, then
$\bT(j\omega) = 0$ and the columns of $\Phi_k$ corresponding to $\bC_k$ can be
safely removed; local parametric modeling is \emph{still needed}, because a single
measurement still provides fewer data equations than unknowns. If it is known a
priori that the equivalent impedance is $dq$-symmetric, the model reduces to
$\bV\!_k = \bG(j\omega_k)\bI_k + \bT(j\omega_k)$ and the columns corresponding
to $\bB_k^-$ can be safely removed; the effect of leakage errors remains, so
local parametric modeling is still needed to eliminate the spectral leakage. If
both hold, an estimate of $\bG$ at $\omega_k$ can be obtained by a single
division $\widehat{\bG(j\omega_k)} = {\bV\!_k}/{\bI_k}$. However, the use of
this simple estimator (a.k.a.\ empirical transfer function estimate, ETFE
\cite[(6.24)]{ljung1998system}) is justified only if steady-state measurements
are possible; otherwise, significant errors would occur
\cite[Lem.~6.1]{ljung1998system}.

\subsection{Extracting the four real TFs}\label{sec:extract}

As pointed out earlier, non-parametric estimates of complex TFs can be used
directly for control and stability analysis. However, if needed, they can be
mapped numerically to non-parametric estimates of $Z_g$. For $\bG(s)$ in
\eqref{eq:single_complex_tf}, writing
$\bG^\ast(s) := [\bG(s^\ast)]^\ast$ for the conjugate reflection, we find that
$G_d(s) = 0.5({\bG(s) + \bG^\ast(s)})$ and
$G_q(s)\! =\! -0.5j({\bG(s) \!-\! \bG^\ast(s)})$. Hence, in the symmetric case
\begingroup\makeatletter\def\f@size{9}\check@mathfonts
\[
\begin{aligned}
 \widehat{Z_{dd}(j\omega_k)}  &= \frac{1}{2} (\widehat{\bG(j\omega_k)} + [\widehat{\bG(j\bar{\omega}_k)}]^\ast),\\
 \widehat{Z_{qd}(j\omega_k)} &= \frac{1}{2j} (\widehat{\bG(j\omega_k)} - [\widehat{\bG(j\bar{\omega}_k)}]^\ast),
\end{aligned}
\]
\endgroup
where $\widehat{\bG(j\bar{\omega}_k)}$ is the estimate of $\bG(-j\omega_{k})$,
$\bar{\omega}_k = \omega_{(N-k)_N}$ and $k \in\{0,\dots, \frac{N}{2}-1\}$
assuming an even $N$. Using a similar reasoning for the asymmetric case,
inverting \eqref{eq:complex_tfs} gives
\begingroup\makeatletter\def\f@size{8}\check@mathfonts
\begin{equation}\label{eq:reconstruction}
\begin{aligned}
\widehat{Z_{dd}(j\omega_k)} &= \frac{1}{2}   \Big( \widehat{\bG_+(j\omega_k)} + [\widehat{\bG_+(j\bar{\omega}_k)}]^\ast + \widehat{\bG_-(j\omega_k)} + [\widehat{\bG_-(j\bar{\omega}_k)}]^\ast \Big),\\
\widehat{Z_{qq}(j\omega_k)} &=  \frac{1}{2}  \Big( \widehat{\bG_+(j\omega_k)} + [\widehat{\bG_+(j\bar{\omega}_k)}]^\ast - \widehat{\bG_-(j\omega_k)} - [\widehat{\bG_-(j\bar{\omega}_k)}]^\ast \Big),       \\
\widehat{Z_{dq}(j\omega_k)} &= \!\frac{-1}{2j} \Big( \widehat{\bG_+(j\omega_k)}\! - [\widehat{\bG_+(j\bar{\omega}_k)}]^\ast\! - \!\widehat{\bG_-(j\omega_k)} + [\widehat{\bG_-(j\bar{\omega}_k)}]^\ast \Big),\\
\widehat{Z_{qd}(j\omega_k)} &= \frac{1}{2j}  \Big( \widehat{\bG_+(j\omega_k)}\! - [\widehat{\bG_+(j\bar{\omega}_k)}]^\ast \!+ \!\widehat{\bG_-(j\omega_k)} - [\widehat{\bG_-(j\bar{\omega}_k)}]^\ast \Big).
\end{aligned}
\end{equation}
\endgroup

The mapping from the complex pair to $Z_g$ alters how estimation errors propagate. 
A zero
of $\bG_+$ is in general not a zero of $Z_{dd}$, because
\eqref{eq:reconstruction} fills the antiresonance with the contribution of the
mirrored line. An error committed at a notch of $\bG_+$ is therefore small in
absolute terms, and it enters an element of $Z_g$ whose magnitude is not small,
so it does not show in an absolute error criterion on $Z_g$. If instead the
complex TFs are used directly, as in SISO sequence-domain or passivity analysis,
the notch has to be resolved on its own, which needs a number of spectral lines
across it.

%% file: sections/measurement.tex
\section{Measurement Chain and Frame Alignment}\label{sec:measurement}

Section~\ref{sec:analysis} treats the $dq$ signals as if they were available
directly. On a real acquisition path the phase quantities first pass through
transducers and through an anti-alias/decimation filter, and only then through
the Park transform. The location of the dynamic filters relative to the Park transform matters,
because a filter with real-valued coefficients in the stationary frame does not act as such in
the $dq$ frame. This section derives the effect and its exact correction, and
then fixes the frame in which the estimate is expressed.

\subsection{A stationary-frame filter seen from the \texorpdfstring{$dq$}{dq} frame}

 For any \emph{real-coefficient} filter $H(s)$ acting in the $abc$
or stationary frame, define its two $dq$ sidebands \cite{Harnefors2007}
\[
H_{+}(\omega) := H(j(\omega_g + \omega)), \quad
H_{-}(\omega) := H(j(\omega_g - \omega)).
\]
Since the coefficients are real, $H(-j\omega) = [{H(j\omega)}]^\ast$. 
\medskip

\subsubsection*{Frame mapping} Let a real filter $H(s)$, with a real impulse response $h(t)$, act on the stationary-frame ($\alpha\beta$) signals and
let the $dq$ signals be obtained by the Park demodulation
$x_{dq}(t) = x_{\alpha\beta}(t)\eu^{-j\omega_g t}$. Then in the $dq$ frame the
filter has the complex impulse response
$h_{dq}(\sigma) = h(\sigma)\eu^{-j\omega_g \sigma}$, equivalently the frequency
response $H_{dq}(j\omega) = H(j(\omega_g+\omega)) = H_{+}(\omega)$.

This can be observed by noticing that with $y_{\alpha\beta} = h \conv x_{\alpha\beta}$ and
$x_{dq}(t) = \eu^{-j\omega_g t}x_{\alpha\beta}(t)$, splitting the exponential as
$\eu^{-j\omega_g t} = \eu^{-j\omega_g\sigma}\eu^{-j\omega_g(t-\sigma)}$ inside the 
convolution gives $y_{dq}(t) = \int [h(\sigma)\eu^{-j\omega_g\sigma}]\,x_{dq}(t-\sigma)\diff\sigma$,
whose transform is $H(j(\omega_g+\omega))$ (see \cite{Harnefors2007} and \cite{Martin2004}).
\smallskip

Next, we analyse what happens when pre-Park transform filters, transducer and anti-aliasing dynamics, are present in the measurement acquisition. 

\subsection{Distortion due to pre-Park filtering}

\begin{proposition}\label{prop:image}
Let the voltage and current acquisition chains have real-coefficient responses
$C^v = S^v H$ and $C^i = S^i H$, a transducer in series with the anti-alias
decimator, applied to the stationary-frame signals, that is, before the Park
transform. At the spectral line $k$ write
$C^{v,+}_{k} := C^v(j(\omega_g+\omega_k))$,
$C^{i,+}_{k} := C^i(j(\omega_g+\omega_k))$ and
$C^{i,-}_{k} := C^i(j(\omega_g-\omega_k))$. Then the estimator applied to the
filtered records returns, at every line $k$,
\begingroup\makeatletter\def\f@size{9.5}\check@mathfonts
\[
\begin{gathered}
\widehat{\Gp(j\omega_k)} = \frac{C^{v,+}_{k}}{C^{i,+}_{k}}\,\Gp(j\omega_k), \quad
\widehat{\Gm(j\omega_k)} = \frac{C^{v,+}_{k}}{[{C^{i,-}_{k}}]^\ast}\,\Gm(j\omega_k).\\[0.2em]
\end{gathered}
\]
\endgroup
\end{proposition}

The proof is given in Appendix~\ref{app:distortion}. 

When the anti-alias/decimation filters are common to the two measurement chains, $\widehat{\Gp} = \Gp$ if and only if
$S^v = S^i$.  Any \emph{transducer mismatch} introduces a proportional bias in $\widehat{\Gp}$ equal to $S^{v,+}/S^{i,+}$. The bias can be corrected if
the transducers are relatively calibrated. Notably, even with matched chains,
$\widehat{\Gm} = D\,\Gm$ with $D = {C^{v,+}_{k}}/{[{C^{i,-}_{k}}]^\ast}$.
For a transducer with a flat response across both sidebands, this factor reduces to the response ratio of $H$ alone.  In what follows we show how to mitigate this distortion.

\subsection{The moving-average decimator}\label{sec:boxcar}

For the moving-average decimator,  the pulse response is $h[n] = 1/N_f$, with $n = 0,\dots,N_f-1$, operating at a rate
$1/T_f$, where $T_f$ is the fast sampling period. Its frequency response,
$H(j2\pi f) = A(f)\eu^{-j2\pi f\tau}$ is  linear-phase. It has a real amplitude response
$A(f) = \sin(\pi f N_f T_f)/(N_f\sin(\pi f T_f))$ and group delay
$\tau = (N_f-1)T_f/2$. The distortion factor according to Proposition \ref{prop:image} is then
\begin{equation}\label{eq:Dfactor}
\begin{aligned}
D(\omega) &= \frac{H(j(\omega_g+\omega))}{[{H(j(\omega_g-\omega))}]^\ast}
= \frac{A(f_g+f)}{A(f_g-f)}\,\eu^{-j2\omega_g\tau}\\
&=: \rho(\omega)\,\eu^{-j2\omega_g\tau}.
\end{aligned}
\end{equation}
The frequency-dependent parts of the two linear phases cancel, leaving a
\emph{constant} phase $-2\omega_g\tau$ at every bin together with a real
magnitude ratio $\rho$, which is the ratio of the passband droop between the two
sidebands. Over an identification band that is narrow compared with the decimation rate, $\rho$ stays close to unity and the associated magnitude factor can be neglected, so the distortion of $\Gm$ is dominated by the constant phase. Whether the distortion can be ignored depends on the parameters of the system and the acquisition channel.

Since $D$ is known at every bin, the corrected estimate
\begin{equation}\label{eq:correction}
\widehat{\Gm}^{\,\mathrm{corr}}(j \omega_k ) = \widehat{\Gm}(j \omega_k)/D(\omega_k),
\end{equation}
for $\Gm$, while $\Gp$ needs no correction (assuming $S^i=S^v$). Notice that the same distortion is less convenient in the real $2\times2$ TF
representation. Splitting $Z = Z_s(\Gp) + Z_a(\Gm)$ into its rotationally invariant
and non-rotationally invariant parts, the pre-Park filter leaves $Z_s$ unchanged and maps $Z_a \mapsto \rho\,R(-2\omega_g\tau)\,Z_a$  where $R$ is a rotation matrix. The combinations
$Z_{dd}+Z_{qq}$ and $Z_{qd}-Z_{dq}$ are therefore clean, but all four entries
individually mix the exact $\Gp$ part with the distorted $\Gm$ part, so an
entry-wise comparison against the true $Z_g$ shows a discrepancy wherever
$|\Gm|$ is non-negligible; it is easily correctable by first decomposing into
$(\Gp,\Gm)$. The conjugate-coupling structure localizes the measurement-chain distortion to a single correctable
channel; this makes the complex-TFs parameterization natural here.

\begin{remark}
 A moving-average filter is common but generally not the preferred choice for anti-aliasing;  a much flatter passband and a sharper transition can be designed
for the same decimation factor. Yet, the moving-average is used here because it is the filter
available on the real-time simulator of Section~\ref{sec:hil}. The above analysis
does not depend on this choice: \eqref{eq:Dfactor} holds for any real-coefficients FIR filter, and
only the numerical values of $\tau$ and $\rho$ change.
\end{remark}
\vspace{-0.5cm}

\subsection{Frame alignment}\label{sec:alignment}
The single angle anticipated in Section~\ref{sec:frames} is obtained as follows.
The $dq$ transform applied to the measured voltage and current records is built
on a synthetic angle $\theta = \int\omega\,\diff t$ with $\omega$ fixed at the
nominal grid frequency. The transformation is thus
not affected by any synchronization dynamics. However, the defined frame has an arbitrary phase reference, and the identified impedance is
rotated relative to the frame aligned with the PCC voltage. The offset is the
steady-state phase difference $\theta_1$ between the synthetic angle and the
converter's PLL angle, and it is obtained by reading the settled angle of the PLL
just before excitation; a type-2 synchronous-reference-frame (SRF) PLL
\cite{chung2000phase} has zero steady-state phase error and therefore reports the
true system-frame angle. The estimate is then aligned by a single static
rotation, under which $\Gp$ is invariant while
$\Gm \mapsto \eu^{j2\theta_1}\Gm$: the rotation changes the phase of $\Gm$ but
not its magnitude. The PLL thus serves only as a one-time angle reader at steady
state, and plays no role during the measurement or the identification.

Notice that the filter phase $-2\omega_g\tau$ of Section~\ref{sec:boxcar} and the
alignment angle $\theta_1$ are both static $dq$-frame rotations acting on $\Gm$,
but they correct different effects, one due to the measurement chain and one to
the choice of frame reference.

%% file: sections/hil.tex
\section{Experimental Validation}\label{sec:hil}

This section presents a validation of the proposed framework using a hardware-in-the-loop system. The identified responses are validated against the analytically derived
small-signal model of the test system, which serves as the reference against which the estimates are assessed.

\subsection{HIL experiment platform} 
The accuracy of the proposed method is experimentally validated using the setup shown in Fig.~\ref{fig:rig}. The main circuit is executed on the ModelingTech MT8020 real-time simulator with a real-time step of $1\,\si{\micro\second}$. The Imperix B-Board 3.0 controller acquires analog signals from the MT8020 and generates PWM gating signals for the converter switches. Both the controller sampling frequency and the converter switching frequency are set to  $10\,\si{\kilo\hertz}$. The output waveforms are recorded using an oscilloscope. 

The MT8020 simulator implements a moving-average FIR filter decimator on its FPGA, before passing the samples to the controller.  Here it decimates from $10^6$ to $10^4$
samples per second. This is the decimator of Section~\ref{sec:boxcar}
with $N_f = 100$, hence $\tau = (N_f-1)T_f/2 = \SI{49.5}{\micro\second}$, a
constant $\Gm$ phase of $-2\omega_g\tau = -1.78^\circ$ at $f_g = \SI{50}{\Hz}$,
and a sideband magnitude ratio $\rho$ within $\pm\SI{0.03}{\dB}$ over the
evaluation band. Because it acts on the stationary-frame signals, ahead of the
Park transform, correction \eqref{eq:correction} applies to 
$\Gm$ on the full signed grid; $\Gp$ requires none, the shared decimator
cancelling exactly by Proposition~\ref{prop:image} as both acquisition chains are identical ($S^v = S^i = 1$). Two static angles therefore rotate $\Gm$: the filter phase $-2\omega_g\tau$, and
the frame-alignment angle $\theta_1$ of Section~\ref{sec:alignment}, read from the
converter PLL at steady state immediately before excitation. Since the grid
frequency is constant and equal to the integrator frequency throughout, this
single reading remains valid.

\begin{figure}[t]
    \centering
    \includegraphics[width=0.85\columnwidth]{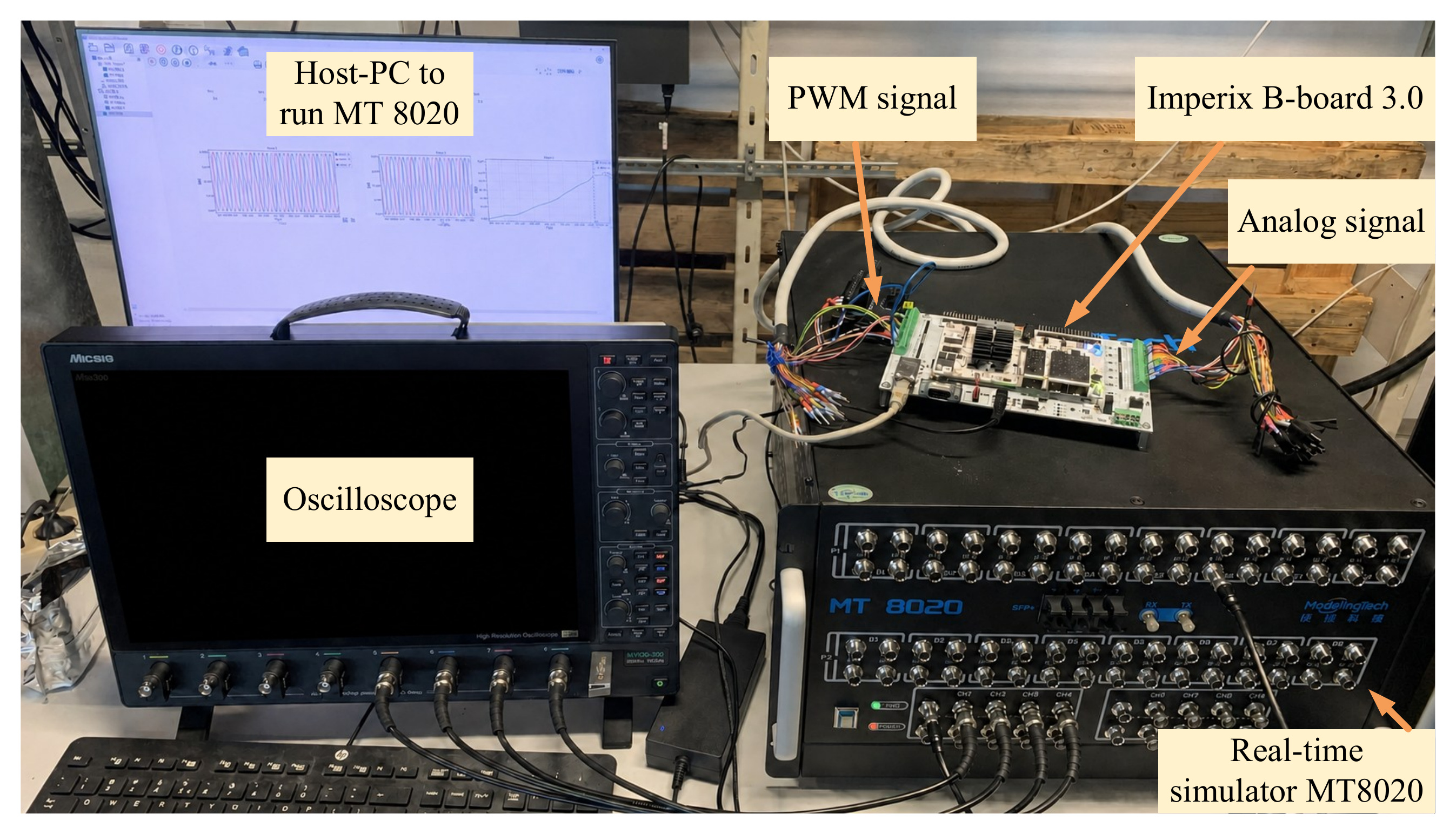}\vspace{-0.15cm}
    \caption{The controller hardware-in-the-loop rig.
    }\vspace{-0.15cm}
    \label{fig:rig}
\end{figure}

\begin{figure*}[t]
    \centering
    \subfloat[Excitation signal at $t_1$.]{%
      \includegraphics[width=0.31\textwidth]{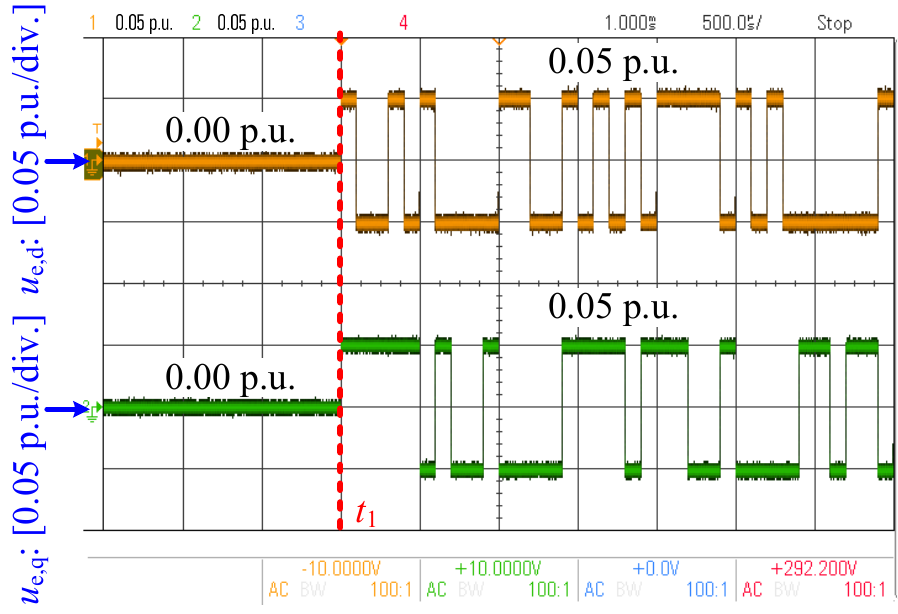}}
    \hfil
    \subfloat[Phase-$a$ voltage.]{%
      \includegraphics[width=0.31\textwidth]{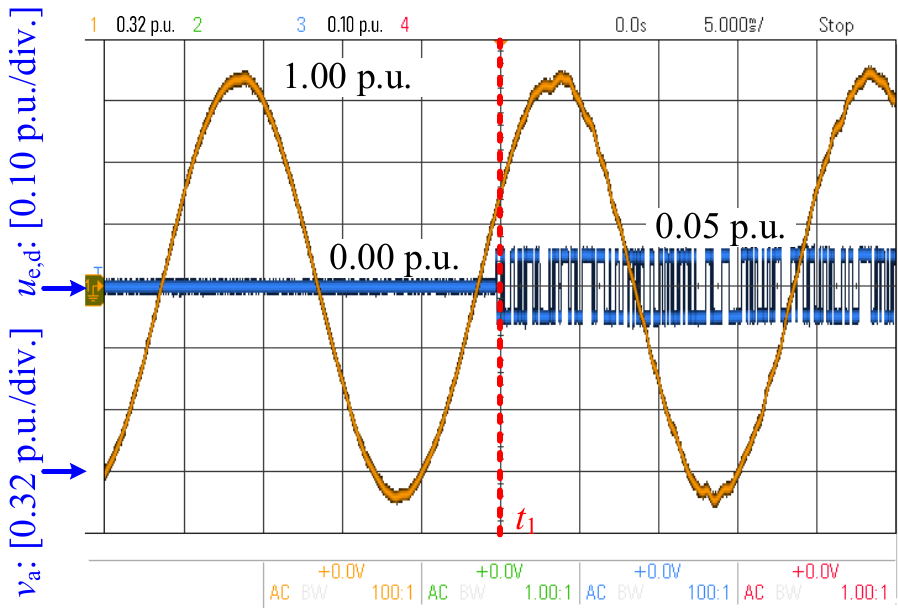}}
    \hfil
    \subfloat[Phase-$a$ current.]{%
      \includegraphics[width=0.31\textwidth]{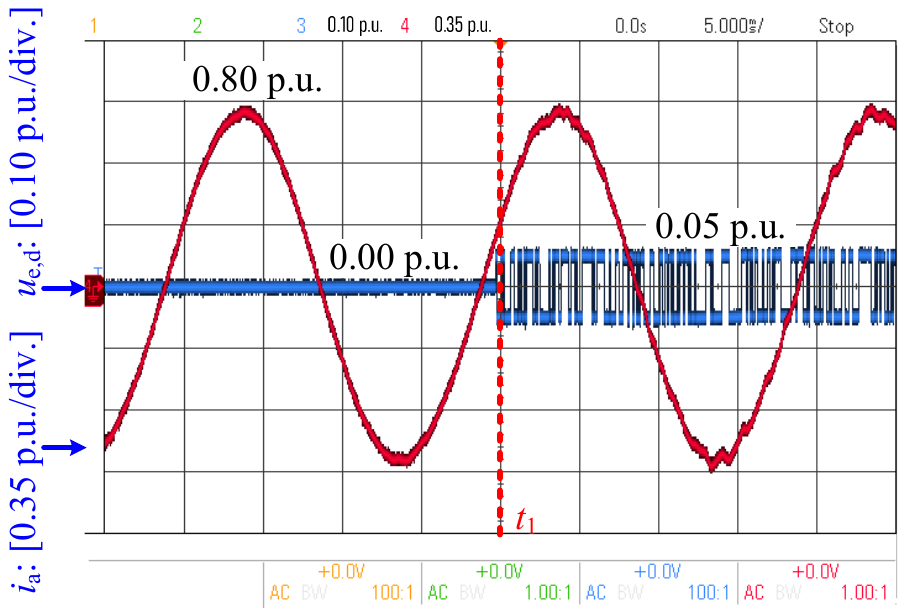}}
    \caption{Oscilloscope traces recorded on the rig. The excitation is a
    zero-mean random binary signal of amplitude $\leq0.05$~p.u.\ added to the
    converter voltage references; $t_1$ marks the instant at which it is
    applied, which is also the instant at which recording begins.}
    \label{fig:scope}
\end{figure*}

The RBS excitation is generated by the controller  and
superimposed on the converter voltage references. Fig.~\ref{fig:scope} shows the
traces recorded on the rig: Fig.~\ref{fig:scope}(a) shows the two excitation
components, $u_{\mathrm{e},\mathrm{d}},u_{\mathrm{e},\mathrm{q}}$,  which are applied at $t_1$ with an amplitude of $0.05$~p.u. Fig.~\ref{fig:scope}(b) and (c)  show the phase-$a$ voltage $v_{\mathrm{a}}$ and
current  $i_{\mathrm{a}}$, respectively, where the superimposed perturbation becomes visible against the fundamental waveform from $t_1$ onward. The excitation amplitude is selected to avoid excessive harmonic distortion ($\text{THD}_v \approx 2.3\%$ based on 10 cycles, all intraharmonics, up to the $50^{\text{th}}$ harmonic). Recording starts with the excitation and runs
for \SI{1}{\second}, giving $N = 10^4$ samples and a DFT resolution of
\SI{1}{\Hz}. One such record is used per case.

The estimator is applied with a local model order $R = 4$ and a local frequency
interval of radius $\ell = 4R+2 = 18$ spectral lines, so that each local problem
uses $2\ell+1 = 37$ lines, to
determine $4R+3 = 19$ complex parameters. The same settings are used for both
cases. Since $37$ lines are available for $19$ parameters, the line at
$\omega = 0$ can be removed, as described in Section~\ref{sec:excitation}, from
every interval that contains it while still leaving enough lines for
identifiability and excitation.

\subsection{Test system configuration}

\begin{figure}[t]
    \centering
    \input{figures/grid_oneline}
    \caption{One-line diagram of the three-phase HIL test system. With the
    breaker open the grid reduces to the passive network of Case~1; with the
    breaker closed it also contains the branch to PCC2 and the grid-following
    VSC of Case~2. }
    \label{fig:grid}
\end{figure}
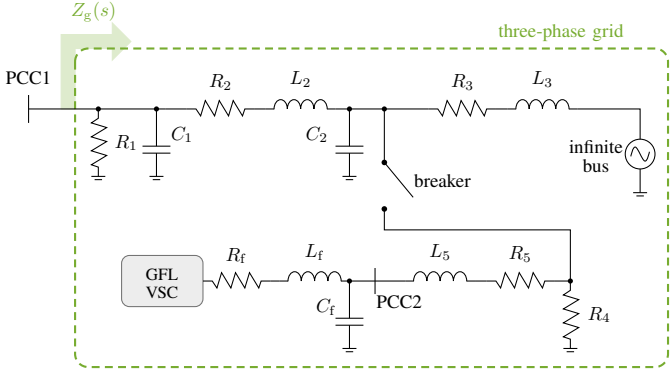

Two cases are considered on the same network.\smallskip

\emph{Case~1} is the passive three-phase grid of Fig.~\ref{fig:grid} with the breaker open. It consists of a resistive load at PCC1, a $\pi$-model transmission line, and an additional $RL$ line connected to an infinite bus. This configuration provides a representative passive-grid benchmark for impedance identification: the series line impedances capture the resistive-inductive characteristics of transmission paths, whereas the shunt capacitances introduce frequency-dependent dynamics and possible resonant behavior beyond that of a simple $RL$ grid equivalent. Its equivalent impedance is $dq$-symmetric.

\emph{Case~2} is qualitatively different from Case~1. With the breaker closed, the passive grid is interconnected with a grid-following voltage-source converter (GFL-VSC) at a second point of common coupling, PCC2, through an $LC$ filter, an additional line, and a shunt load. The breaker connects this branch to the main line at the node between the second shunt capacitor $C_2$ and the line to the infinite bus. Unlike Case~1, this system therefore contains an actively controlled power-electronic device, introducing additional control dynamics. The GFL-VSC is supplied by an ideal DC link and employs an inner current loop, outer active- and reactive-power loops, an SRF-PLL \cite{chung2000phase}, and feedforward decoupling \cite{teodorescu2011grid}; these control dynamics render the equivalent impedance $dq$-asymmetric. More importantly, this case represents a practical scenario in which a transmission system operator identifies the dynamic terminal model of a converter-based resource from grid interface measurements. Table~\ref{tab:params} lists the parameters of both cases.

\input{figures/fig_ratio}

Fig.~\ref{fig:ratio} quantifies how much the coupling matters in Case 2, as the pointwise
ratio $r(\omega) = |\Gm(j\omega)|/|\Gp(j\omega)|$: about $\SI{-10}{\dB}$ at the
lines adjacent to zero, of the order of $\SI{-28}{\dB}$ within
$\pm\SI{20}{\Hz}$, between $\SI{-37}{\dB}$ and $\SI{-42}{\dB}$ over $20$ to
$\SI{200}{\Hz}$, and near $\SI{-54}{\dB}$ beyond $\pm\SI{200}{\Hz}$. Neglecting
cross-coupling is therefore defensible at high frequencies away from the
fundamental, but not close to it. This is where subsynchronous and PLL-driven
interactions occur, and where the sharpest feature of $\Gp$ is also located, so
that resolution matters most. 

\input{sections/table_params}

\subsection{Case 1: a \texorpdfstring{$dq$}{dq}-symmetric grid}

\input{figures/fig_case1_Gplus}
\input{figures/fig_case1_Gminus}
\input{figures/fig_case1_Z}

Figs.~\ref{fig:c1Gp} and \ref{fig:c1Gm} show the identified responses of the
complex TFs over $[-1,1]\,\si{\kHz}$ and Fig.~\ref{fig:c1Z} the responses of the
four real TFs over $[1,1200]\,\si{\Hz}$, against the analytical model.  This frequency range remains well below the Nyquist frequency. The estimate of $\Gp$
overlays the model across the band, including the antiresonance at
$\SI{-50}{\Hz}$, where $|\Gp|$ falls to $\SI{-20.8}{\dB}$. 
 That notch is the image of the $abc$ direct-current point:
the $dq$ frequency axis is the $abc$ axis shifted by $-\omega_g$, so
$\omega_p = -\omega_g$ maps to zero frequency in phase coordinates, where the
passive network reduces to its resistive divider
$R_1\!\parallel\!(R_2+R_3) = 0.0907$~p.u. Its \SI{3}{\dB} width is
\SI{17.2}{\Hz}, so about seventeen spectral lines fall across it at \SI{1}{\Hz}
resolution, against a local interval of \SI{37}{\Hz}. This reduces the interpolation errors significantly.  

The estimator does not make use of the knowledge that the grid is
$dq$-symmetric in this case, and identifies both TFs throughout. Its estimate of $\Gm$,
whose true value is identically zero, sits at the identification noise floor,
with a median of \SI{-77}{\dB} over the evaluation band against a peak $|\Gp|$
of \SI{6.0}{\dB}, a separation of some \SI{80}{\dB}. The absence of coupling
is thus obtained from the data and not assumed.

\subsection{Case 2: a \texorpdfstring{$dq$}{dq}-asymmetric grid}

\input{figures/fig_case2_Gplus}
\input{figures/fig_case2_Gminus}
\input{figures/fig_case2_Z}
\input{figures/fig_zoom}

With the breaker closed, the added converter changes $\Gp$ and makes the
equivalent impedance asymmetric, so $\Gm$ ceases to be zero. Figs.~\ref{fig:c2Gp} and \ref{fig:c2Gm} show the responses of both complex TFs
against the model and Fig.~\ref{fig:c2Z} those of the four real TFs. 
The direct TF $\Gp$ retains the $\SI{-50}{\Hz}$ notch, now \SI{26.1}{\dB} deep with a \SI{3}{\dB} width of \SI{17.6}{\Hz},  and gets a new peak around $\SI{0}{\Hz}$.

The coupling TF $\Gm$ has a distinctive shape. It peaks at \SI{-16.9}{\dB} within a
hertz of the fundamental and drops into minima of \SI{-55.4}{\dB} and
\SI{-54.5}{\dB} at $\mp\SI{50}{\Hz}$. It then recovers to local maxima near \SI{-43}{\dB} at about $\mp\SI{160}{\Hz}$, and decays to roughly \SI{-56}{\dB} at $\pm\SI{1}{\kHz}$.
The estimate tracks the model over this whole range. Below roughly \SI{-50}{\dB} a scatter of a few decibels appears and grows towards the band edges, where the coupling channel becomes limited by SNR.

Fig.~\ref{fig:zoom} examines the three sharpest features
directly: the $\SI{-50}{\Hz}$ notch of $\Gp$, the near-fundamental peak of
$\Gm$, which rises by nearly \SI{20}{\dB} within $\pm\SI{8}{\Hz}$ of zero, and
the $\SI{-50}{\Hz}$ minimum of $\Gm$. All three are resolved at \SI{1}{\Hz},
although only a handful of lines fall across the peak of $\Gm$, which
is the situation anticipated in Section~\ref{sec:extract}.

\subsection{Accuracy}

\input{sections/table_metrics}

To evaluate the performance of the proposed method, we use the following two metrics. For any real
response, with
$Z_{0:k} := \begin{bmatrix} Z(j\omega_0) & \cdots & Z(j\omega_k)\end{bmatrix}^\top$
the true values and $\widehat{Z}_{0:k}$ their estimate,
\begin{equation}\label{eq:fit}
\mathrm{Fit\%} := \left(1 - \frac{\|\widehat{Z}_{0:k} - Z_{0:k}\|_2^2}{\|Z_{0:k} - \mathrm{mean}(Z_{0:k})\|_2^2}\right)\times 100,
\end{equation}
with $\mathrm{mean}(Z_{0:k}) = \frac{1}{k+1}\sum_{n=0}^k Z(j\omega_n)$. Fit\%
may be negative; larger is better and a perfect estimate scores $100$. The
second metric is the worst-case relative error
\begin{equation}\label{eq:hinf}
\text{relative } H_\infty \text{ error} := \frac{\max_{n} \sigmax\big(\widehat{Z_g}(j\omega_n) - Z_g(j\omega_n)\big)}{\max_{n} \sigmax\big(Z_g(j\omega_n)\big)},
\end{equation}
with $\sigmax$ the largest singular value. The two metrics are complementary.
Fit\% is an aggregate measure: it compares the energy of the estimation error
over the whole band with the energy of the variation of the true response about
its mean, so it reports how well the overall shape of the frequency response is
captured. The relative $H_\infty$ error is a worst-case measure: it reports the
largest discrepancy at any single frequency, normalized by the largest gain of
the true response, and is therefore sensitive to localized errors, for instance
close to a resonance, which the aggregate measure would average out. Both extend
directly to the complex responses on the two-sided grid: \eqref{eq:fit} is unchanged, and
for a scalar complex response, $\sigmax$ reduces to the modulus, so
\eqref{eq:hinf} becomes $\max_k|\widehat{\bG_k} - \bG_k|/\max_k|\bG_k|$ over the band. The
complex responses are evaluated over $[-1,1]\,\si{\kHz}$ and the real ones over
$[1,1200]\,\si{\Hz}$, the same bands as the corresponding figures.

Table~\ref{tab:metrics} collects the results. In Case~1 all four real responses
are fitted to within rounding of $100\%$ and the worst-case relative error is
$0.4\%$. In Case~2 the fits remain above $99.7\%$ and the worst-case relative
error is $6.5\%$, with the largest error at \SI{1}{\Hz}, i.e.\ at the sharp near-fundamental
feature. The complex responses tell the same
story: $\Gp$ is recovered essentially exactly in both cases, while the $85.57\%$
fit of $\Gm$ in Case~2 is dominated by that same narrow peak, whose amplitude
changes by more than a decibel between adjacent lines.

\subsection{Local model order}\label{sec:order}

All results reported above use a local model order $R = 4$ with
$\ell = 4R+2 = 18$. Table~\ref{tab:order} repeats the accuracy metrics for
$R = 4, 6, 8$ and $10$, with $\ell$ adjusted accordingly. The estimates are fairly
insensitive to this choice: over the whole range the Fit\% of every real TF
varies by less than $0.03$ in Case~1 and by less than $0.04$ in Case~2, and the
relative $H_\infty$ error changes by less than $0.01$. This is the behaviour
expected of local rational modeling, and it contrasts with global parametric
methods, whose accuracy depends strongly on the selected order. 

Table~\ref{tab:ordercomplex} reports the same two metrics for the complex TFs.
The estimate of $\bG_+$ is essentially exact at every order in both cases, while
the figures for $\bG_-$ in Case~2 are visibly poorer and show slight improvement with larger $R$.  This is due to the small magnitude of $\Gm$, which is three orders of magnitude
smaller than $\Gp$ around \SI{0}{\Hz}, and negligible at higher frequencies. 

\input{sections/table_order}
\input{sections/table_order_complex}

%% file: figures/grid_oneline.tex
\resizebox{\columnwidth}{!}{%
\begin{tikzpicture}[circuit ee IEC, scale=0.62, every node/.style={scale=0.72}]

\draw[fill=greenBackgroundcolor!30,color=greenBackgroundcolor!30,opacity=0.8]
  (1.15,0) -- (1.15,1.60) -- (2.25,1.60) -- (2.25,1.85) -- (2.65,1.45)
  -- (2.25,1.05) -- (2.25,1.30) -- (1.45,1.30) -- (1.45,0) -- (1.15,0);
\node[color=greenBackgroundcolor] at (1.85,2.10) {$Z_\mathrm{g}(s)$};

\draw[densely dashed, rounded corners=5, color=greenBackgroundcolor, thick]
  (1.45,1.30) rectangle (14.2,-5.55);
\node[color=greenBackgroundcolor] at (11.9,1.65) {three-phase grid};

\draw (0.45,0.30) -- (0.45,-0.30);
\node at (0.45,0.68) {PCC1};
\draw (0.45,0) -- (1.45,0);

\draw (1.45,0) -- (3.40,0);
\fill (1.95,0) circle (0.055);
\draw (1.95,0) \resV;
\gnd{1.95}{-1.45}
\node at (2.55,-0.72) {$R_1$};
\fill (3.20,0) circle (0.055);
\shuntcap{3.20}{0}
\node at (3.75,-0.55) {$C_1$};

\draw (3.40,0) -- (3.85,0);
\draw (3.85,0) \resH;
\draw (5.30,0) -- (5.55,0);
\draw (5.55,0) to [inductor={yscale=1.4}] (7.05,0);
\node at (4.58,0.58) {$R_2$};
\node at (6.30,0.65) {$L_2$};

\draw (7.05,0) -- (9.05,0);
\fill (8.10,0) circle (0.055);
\fill (7.35,0) circle (0.055);
\shuntcap{7.35}{0}
\node at (6.65,-0.55) {$C_2$};

\draw (9.05,0) \resH;
\draw (10.50,0) -- (10.75,0);
\draw (10.75,0) to [inductor={yscale=1.4}] (12.25,0);
\node at (9.78,0.58) {$R_3$};
\node at (11.50,0.65) {$L_3$};

\draw (12.25,0) -- (13.60,0) -- (13.60,-0.65);
\acsource{13.60}{-0.97}
\draw (13.60,-1.29) -- (13.60,-1.80);
\gnd{13.60}{-1.80}
\node at (12.65,-0.80) {infinite};
\node at (12.65,-1.22) {bus};

\draw (8.10,0) -- (8.10,-1.15);
\fill (8.10,-1.15) circle (0.06);
\fill (8.10,-2.15) circle (0.06);
\draw (8.10,-1.15) -- (8.72,-1.90);
\node at (9.35,-1.55) {breaker};
\draw (8.10,-2.15) -- (8.10,-2.60) -- (12.10,-2.60) -- (12.10,-3.70);

\draw[rounded corners=3, fill=black!7, draw=black!55] (2.45,-4.25) rectangle (4.20,-3.15);
\node at (3.32,-3.55) {\small GFL};
\node at (3.32,-3.95) {\small VSC};

\draw (4.20,-3.70) \resH;
\draw (5.65,-3.70) -- (5.85,-3.70);
\node at (4.92,-3.15) {$R_\mathrm{f}$};
\draw (5.85,-3.70) to [inductor={yscale=1.4}] (7.35,-3.70);
\node at (6.60,-3.08) {$L_\mathrm{f}$};
\fill (7.35,-3.70) circle (0.055);
\shuntcap{7.35}{-3.70}
\node at (6.85,-4.25) {$C_\mathrm{f}$};

\draw (7.90,-3.42) -- (7.90,-3.98);
\node at (8.42,-4.15) {PCC2};

\draw (7.35,-3.70) -- (8.55,-3.70);
\draw (8.55,-3.70) to [inductor={yscale=1.4}] (10.05,-3.70);
\node at (9.30,-3.08) {$L_5$};
\draw (10.05,-3.70) -- (10.30,-3.70);
\draw (10.30,-3.70) \resH;
\node at (11.02,-3.12) {$R_5$};
\draw (11.75,-3.70) -- (12.10,-3.70);

\fill (12.10,-3.70) circle (0.055);
\draw (12.10,-3.70) \resV;
\gnd{12.10}{-5.15}
\node at (12.72,-4.42) {$R_4$};

\end{tikzpicture}}

%% file: figures/fig_ratio.tex
\begin{figure}[t]
\centering
\begin{tikzpicture}
\begin{axis}[hilaxis, width=1.02\columnwidth, height=4.0cm,
             xmin=-1000, xmax=1000, xtick={-1000,-500,0,500,1000},
             ymin=-75, ymax=0, ytick={0,-20,-40,-60},
             xlabel={Frequency [\si{\Hz}]}, ylabel={$r(\omega)$ [\si{\dB}]}]
\addplot[truecurve] table[x=f,y=r_dB] {figures/data/case2_ratio_true.dat};
\end{axis}
\end{tikzpicture}
\caption{Case~2: the ratio $r(\omega)=|\bG_-(j\omega)|/|\bG_+(j\omega)|$, which
measures how strong the mirror-frequency coupling is relative to the direct
term. The coupling is significant only in a narrow band around the
fundamental.}
\label{fig:ratio}
\end{figure}
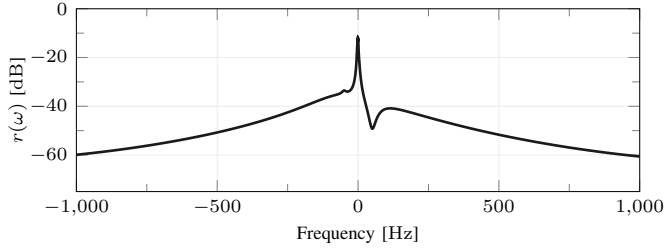

%% file: sections/table_params.tex
\begin{table}[t]
\renewcommand{\arraystretch}{1.08}
\caption{Parameters of the HIL test system. Values in per unit unless stated.}
\label{tab:params}
\centering
\footnotesize
\setlength{\tabcolsep}{4pt}
\begin{tabular}{@{}p{3.85cm}p{1.85cm}l@{}}
\toprule
Parameter & Symbol & Value \\
\midrule
\multicolumn{3}{@{}l}{\textit{Base quantities and platform}}\\
Voltage, power \& freq.\ base & $V_\mathrm{b},S_\mathrm{b},f_\mathrm{b}$ & \SI{380}{\volt}, 1\,MVA, \SI{50}{\Hz}\\
Control step / switching freq. & $T_s$ & \SI{e-4}{\second} / \SI{10}{\kHz}\\
FPGA step, decimation length & $T_f$, $N_f$ & 1\si{\micro\second}, 100\\
\midrule
\multicolumn{3}{@{}l}{\textit{Grid (both cases)}}\\
Load at PCC1 & $R_1$ & 2\\
Shunt capacitors & $C_1, C_2$ & 0.005, 0.005\\
Line 1 & $R_2, L_2$ & 0.015, 0.15\\
Grid line & $R_3, L_3$ & 0.08, 0.4\\
\midrule
\multicolumn{3}{@{}l}{\textit{Excitation converter (both cases)}}\\
$LCL$ filter & $L_\mathrm{f,1}, L_\mathrm{f,2}, C_\mathrm{f}$ & 0.08, 0.05, 0.02\\
$LCL$ resistances, damping & $r_1, r_2, r_\mathrm{d}$ & 0.008, 0.005, 0.4\\
DC link voltage & $v_\mathrm{dc}$ & \SI{1150}{\volt}\\
Current PI gains & $k_p, k_i$ & 0.3, 10\\
$P$ and $Q$ PI gains & $k_p, k_i$ & 0.5, 40 (each)\\
PLL PI gains& $k_p, k_i$ & 6.87, 23.61\\
Set point, excitation ampl. & $P, Q$ & 0.8, 0; $\leq 0.05$\\
\midrule
\multicolumn{3}{@{}l}{\textit{Added grid-following converter (Case~2 only)}}\\
Shunt load at breaker & $R_4$ & 2\\
Line 2 & $R_5, L_5$ & 0.015, 0.15\\
$LC$ filter & $R_\mathrm{f}, L_\mathrm{f}, C_\mathrm{f}$ & 0.01, 0.05, 0.02\\
DC link voltage & $v_\mathrm{dc}$ & \SI{931}{\volt}\\
Current PI gains & $k_p, k_i$ & 0.3, 10\\
$P$ and $Q$ PI gains & $k_p, k_i$ & 0.1, 5 (each)\\
PLL PI gains & $k_p, k_i$ & 6.87, 23.61\\
Set point & $P, Q$ & 0.95, 0.2\\
\midrule
\multicolumn{3}{@{}l}{Equivalent impedance $Z_g(s)$, Case~1: $dq$-symmetric, rational $2\!\times\!2$, order 8}\\
\bottomrule
\end{tabular}
\end{table}

%% file: figures/fig_case1_Gplus.tex
\begin{figure}[t]
\centering
\begin{tikzpicture}
\begin{groupplot}[group style={group size=1 by 2, vertical sep=1.05cm},
                   hilaxis, width=1.02\columnwidth, height=3.6cm,
                   xmin=-1000, xmax=1000, xtick={-1000,-500,0,500,1000}]
\nextgroupplot[ylabel={$|\bG_+|$ [\si{\dB}]}, ymin=-25, ymax=10, ytick={0,-10,-20},
               xticklabels={,,}]
\addplot[truecurve] table[x=f,y=Gp_mag] {figures/data/case1_Gtrue.dat};
\addplot[estmarks] table[x=f,y=Gp_mag] {figures/data/case1_Ghat.dat};
\legend{True, Proposed approach ($R=4$)}
\nextgroupplot[ylabel={$\angle \bG_+$ [degrees]}, ymin=-180, ymax=180,
               ytick={-180,-90,0,90,180}, xlabel={Frequency [\si{\Hz}]}]
\addplot[truecurve] table[x=f,y=Gp_ph] {figures/data/case1_Gtrue.dat};
\addplot[estmarks] table[x=f,y=Gp_ph] {figures/data/case1_Ghat.dat};
\end{groupplot}
\end{tikzpicture}
\caption{Case~1: magnitude and phase of the complex TF $\bG_+$ and an estimate obtained
using the proposed approach via a local model order 4, from a single one-second
record (for clarity, only each 5th estimated frequency is shown). The
antiresonance at $-\SI{50}{\Hz}$ is the image of the $abc$ direct-current
point.}
\label{fig:c1Gp}
\end{figure}
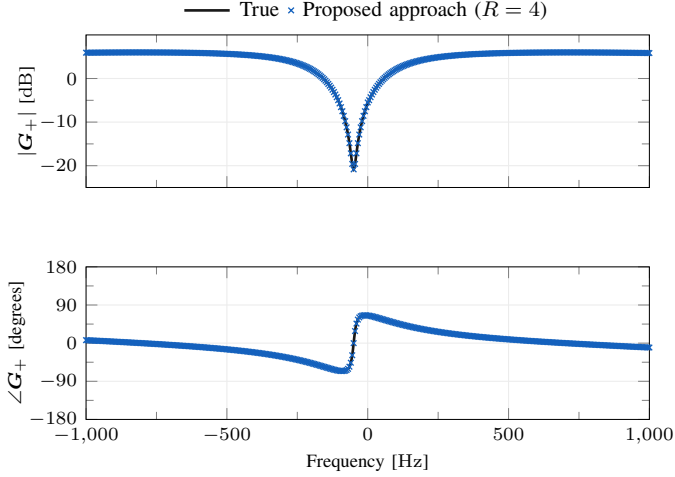

%% file: figures/fig_case1_Gminus.tex
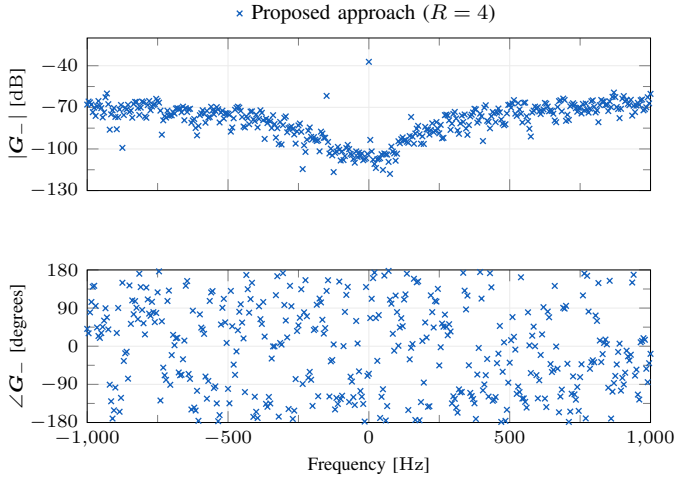
\begin{figure}[t]
\centering
\begin{tikzpicture}
\begin{groupplot}[group style={group size=1 by 2, vertical sep=1.05cm},
                   hilaxis, width=1.02\columnwidth, height=3.6cm,
                   xmin=-1000, xmax=1000, xtick={-1000,-500,0,500,1000}]
\nextgroupplot[ylabel={$|\bG_-|$ [\si{\dB}]}, ymin=-130, ymax=-20, ytick={-40,-70,-100,-130},
               xticklabels={,,}]
\addplot[estmarks] table[x=f,y=Gm_mag] {figures/data/case1_Ghat.dat};
\legend{Proposed approach ($R=4$)}
\nextgroupplot[ylabel={$\angle \bG_-$ [degrees]}, ymin=-180, ymax=180,
               ytick={-180,-90,0,90,180}, xlabel={Frequency [\si{\Hz}]}]
\addplot[estmarks] table[x=f,y=Gm_ph] {figures/data/case1_Ghat.dat};
\end{groupplot}
\end{tikzpicture}
\caption{Case~1: magnitude and phase of the estimate of the coupling TF $\bG_-$. The true
value is zero, and the estimate reflects the noise floor of the identification;
its phase is correspondingly arbitrary. The estimator does not make use of the knowledge
that this grid is $dq$-symmetric; it rather discovers this from the data.}
\label{fig:c1Gm}
\end{figure}

%% file: figures/fig_case1_Z.tex
\begin{figure*}[t]
\centering
\begin{tikzpicture}
\begin{groupplot}[group style={group size=4 by 2, horizontal sep=1.15cm, vertical sep=0.95cm},
                   hilaxis, width=0.265\textwidth, height=3.1cm,
                   xmode=log, xmin=1, xmax=1200, xtick={1,10,100,1000},
                   title style={font=\small, yshift=-0.6ex},
                   legend style={at={(0.02,1.30)}, anchor=south west, legend columns=2}]
\nextgroupplot[ylabel={Magnitude [\si{\dB}]}, ymin=-20, ymax=10, ytick={0,-10,-20},
               title={$\mathstrut Z_{dd}$}, xticklabels={,,}]
\addplot[truecurve] table[x=f,y=dd_mag] {figures/data/case1_Ztrue.dat};
\addplot[estmarks] table[x=f,y=dd_mag] {figures/data/case1_Zhat.dat};
\legend{True, Proposed approach ($R=4$)}
\nextgroupplot[ ymin=-35, ymax=0, ytick={0,-10,-20,-30},
               title={$\mathstrut Z_{dq}$}, xticklabels={,,}]
\addplot[truecurve] table[x=f,y=dq_mag] {figures/data/case1_Ztrue.dat};
\addplot[estmarks] table[x=f,y=dq_mag] {figures/data/case1_Zhat.dat};
\nextgroupplot[ ymin=-35, ymax=0, ytick={0,-10,-20,-30},
               title={$\mathstrut Z_{qd}$}, xticklabels={,,}]
\addplot[truecurve] table[x=f,y=qd_mag] {figures/data/case1_Ztrue.dat};
\addplot[estmarks] table[x=f,y=qd_mag] {figures/data/case1_Zhat.dat};
\nextgroupplot[ ymin=-20, ymax=10, ytick={0,-10,-20},
               title={$\mathstrut Z_{qq}$}, xticklabels={,,}]
\addplot[truecurve] table[x=f,y=qq_mag] {figures/data/case1_Ztrue.dat};
\addplot[estmarks] table[x=f,y=qq_mag] {figures/data/case1_Zhat.dat};
\nextgroupplot[ylabel={Phase [degrees]}, ymin=-30, ymax=60, ytick={-30,0,30,60},
               xlabel={Frequency [\si{\Hz}]}]
\addplot[truecurve] table[x=f,y=dd_ph] {figures/data/case1_Ztrue.dat};
\addplot[estmarks] table[x=f,y=dd_ph] {figures/data/case1_Zhat.dat};
\nextgroupplot[ ymin=-45, ymax=190, ytick={-45,0,90,180},
               xlabel={Frequency [\si{\Hz}]}]
\addplot[truecurve] table[x=f,y=dq_ph] {figures/data/case1_Ztrue.dat};
\addplot[estmarks] table[x=f,y=dq_ph] {figures/data/case1_Zhat.dat};
\nextgroupplot[ ymin=-225, ymax=15, ytick={-180,-90,0},
               xlabel={Frequency [\si{\Hz}]}]
\addplot[truecurve] table[x=f,y=qd_ph] {figures/data/case1_Ztrue.dat};
\addplot[estmarks] table[x=f,y=qd_ph] {figures/data/case1_Zhat.dat};
\nextgroupplot[ ymin=-30, ymax=60, ytick={-30,0,30,60},
               xlabel={Frequency [\si{\Hz}]}]
\addplot[truecurve] table[x=f,y=qq_ph] {figures/data/case1_Ztrue.dat};
\addplot[estmarks] table[x=f,y=qq_ph] {figures/data/case1_Zhat.dat};
\end{groupplot}
\end{tikzpicture}
\caption{Case~1: magnitude and phase of the four real TFs, mapped from the complex
estimates via \eqref{eq:reconstruction}, together with the true values, over
$[1,1200]\,\si{\Hz}$ (for clarity, only each 5th estimated frequency is
shown).}
\label{fig:c1Z}
\end{figure*}
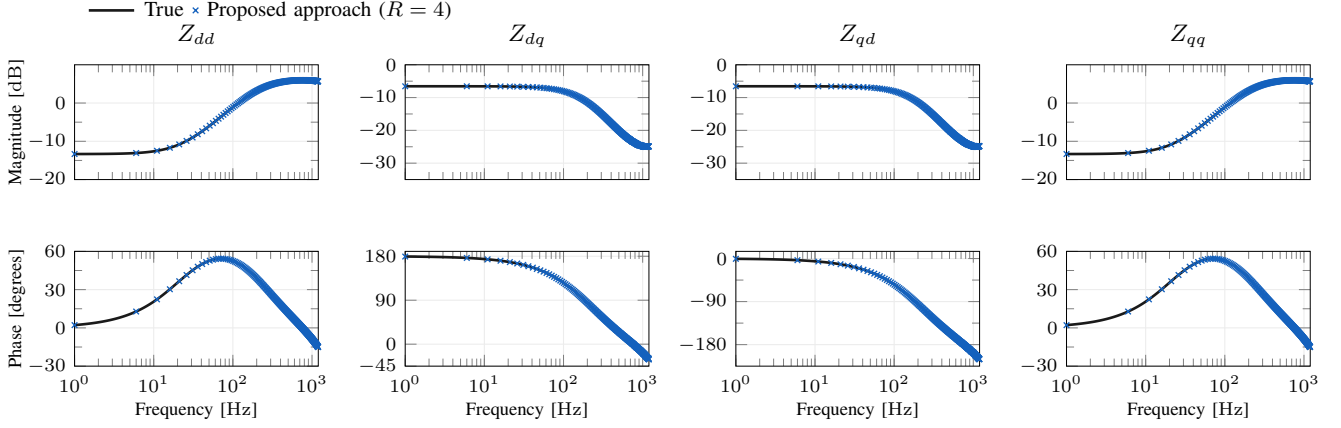

%% file: figures/fig_case2_Gplus.tex
\begin{figure}[t]
\centering
\begin{tikzpicture}
\begin{groupplot}[group style={group size=1 by 2, vertical sep=1.05cm},
                   hilaxis, width=1.02\columnwidth, height=3.6cm,
                   xmin=-1000, xmax=1000, xtick={-1000,-500,0,500,1000}]
\nextgroupplot[ylabel={$|\bG_+|$ [\si{\dB}]}, ymin=-25, ymax=10, ytick={0,-10,-20},
               xticklabels={,,}]
\addplot[truecurve] table[x=f,y=Gp_mag] {figures/data/case2_Gtrue.dat};
\addplot[estmarks] table[x=f,y=Gp_mag] {figures/data/case2_Ghat.dat};
\legend{True, Proposed approach ($R=4$)}
\nextgroupplot[ylabel={$\angle \bG_+$ [degrees]}, ymin=-180, ymax=180,
               ytick={-180,-90,0,90,180}, xlabel={Frequency [\si{\Hz}]}]
\addplot[truecurve] table[x=f,y=Gp_ph] {figures/data/case2_Gtrue.dat};
\addplot[estmarks] table[x=f,y=Gp_ph] {figures/data/case2_Ghat.dat};
\end{groupplot}
\end{tikzpicture}
\caption{Case~2: magnitude and phase of $\bG_+$ and an estimate obtained using the
proposed approach via a local model order 4, with the added grid-following
converter in service (for clarity, only each 5th estimated frequency is
shown).}
\label{fig:c2Gp}
\end{figure}
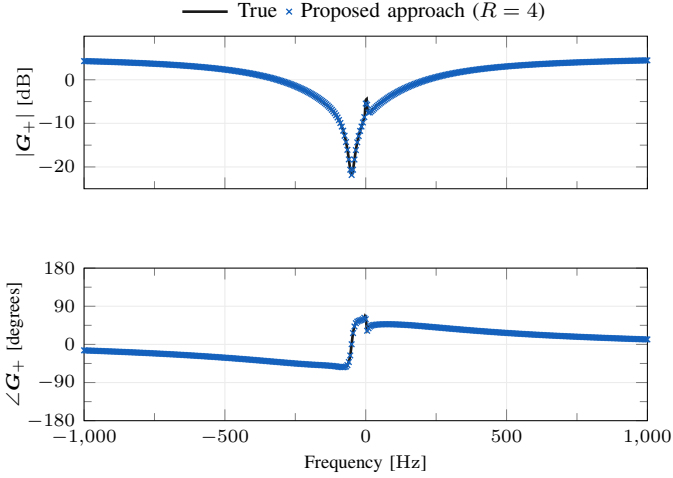

%% file: figures/fig_case2_Gminus.tex
\begin{figure}[t]
\centering
\begin{tikzpicture}
\begin{groupplot}[group style={group size=1 by 2, vertical sep=1.05cm},
                   hilaxis, width=1.02\columnwidth, height=3.6cm,
                   xmin=-1000, xmax=1000, xtick={-1000,-500,0,500,1000}]
\nextgroupplot[ylabel={$|\bG_-|$ [\si{\dB}]}, ymin=-70, ymax=-10, ytick={-10,-30,-50,-70},
               xticklabels={,,}]
\addplot[truecurve] table[x=f,y=Gm_mag] {figures/data/case2_Gtrue.dat};
\addplot[estmarks] table[x=f,y=Gm_mag] {figures/data/case2_Ghat.dat};
\legend{True, Proposed approach ($R=4$)}
\nextgroupplot[ylabel={$\angle \bG_-$ [degrees]}, ymin=-180, ymax=180,
               ytick={-180,-90,0,90,180}, xlabel={Frequency [\si{\Hz}]}]
\addplot[truecurve] table[x=f,y=Gm_ph] {figures/data/case2_Gtrue.dat};
\addplot[estmarks] table[x=f,y=Gm_ph] {figures/data/case2_Ghat.dat};
\end{groupplot}
\end{tikzpicture}
\caption{Case~2: magnitude and phase of the coupling TF $\bG_-$ and its estimate. The
estimate is corrected for the anti-alias decimator using \eqref{eq:correction}
on the full signed frequency grid (for clarity, only each 5th estimated
frequency is shown).}
\label{fig:c2Gm}
\end{figure}
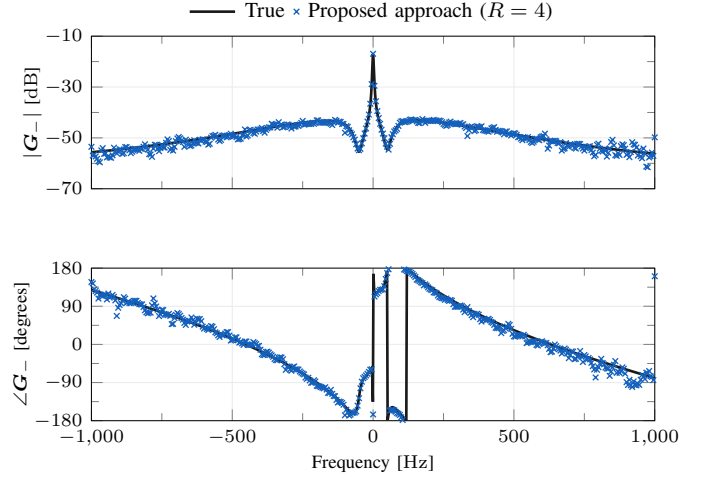

%% file: figures/fig_case2_Z.tex
\begin{figure*}[t]
\centering
\begin{tikzpicture}
\begin{groupplot}[group style={group size=4 by 2, horizontal sep=1.15cm, vertical sep=0.95cm},
                   hilaxis, width=0.265\textwidth, height=3.1cm,
                   xmode=log, xmin=1, xmax=1200, xtick={1,10,100,1000},
                   title style={font=\small, yshift=-0.6ex},
                   legend style={at={(0.02,1.30)}, anchor=south west, legend columns=2}]
\nextgroupplot[ylabel={Magnitude [\si{\dB}]}, ymin=-20, ymax=10, ytick={0,-10,-20},
               title={$\mathstrut Z_{dd}$}, xticklabels={,,}]
\addplot[truecurve] table[x=f,y=dd_mag] {figures/data/case2_Ztrue.dat};
\addplot[estmarks] table[x=f,y=dd_mag] {figures/data/case2_Zhat.dat};
\legend{True, Proposed approach ($R=4$)}
\nextgroupplot[ ymin=-35, ymax=0, ytick={0,-10,-20,-30},
               title={$\mathstrut Z_{dq}$}, xticklabels={,,}]
\addplot[truecurve] table[x=f,y=dq_mag] {figures/data/case2_Ztrue.dat};
\addplot[estmarks] table[x=f,y=dq_mag] {figures/data/case2_Zhat.dat};
\nextgroupplot[ ymin=-35, ymax=0, ytick={0,-10,-20,-30},
               title={$\mathstrut Z_{qd}$}, xticklabels={,,}]
\addplot[truecurve] table[x=f,y=qd_mag] {figures/data/case2_Ztrue.dat};
\addplot[estmarks] table[x=f,y=qd_mag] {figures/data/case2_Zhat.dat};
\nextgroupplot[ ymin=-20, ymax=10, ytick={0,-10,-20},
               title={$\mathstrut Z_{qq}$}, xticklabels={,,}]
\addplot[truecurve] table[x=f,y=qq_mag] {figures/data/case2_Ztrue.dat};
\addplot[estmarks] table[x=f,y=qq_mag] {figures/data/case2_Zhat.dat};
\nextgroupplot[ylabel={Phase [degrees]}, ymin=-30, ymax=60, ytick={-30,0,30,60},
               xlabel={Frequency [\si{\Hz}]}]
\addplot[truecurve] table[x=f,y=dd_ph] {figures/data/case2_Ztrue.dat};
\addplot[estmarks] table[x=f,y=dd_ph] {figures/data/case2_Zhat.dat};
\nextgroupplot[ ymin=-45, ymax=190, ytick={-45,0,90,180},
               xlabel={Frequency [\si{\Hz}]}]
\addplot[truecurve] table[x=f,y=dq_ph] {figures/data/case2_Ztrue.dat};
\addplot[estmarks] table[x=f,y=dq_ph] {figures/data/case2_Zhat.dat};
\nextgroupplot[ ymin=-225, ymax=15, ytick={-180,-90,0},
               xlabel={Frequency [\si{\Hz}]}]
\addplot[truecurve] table[x=f,y=qd_ph] {figures/data/case2_Ztrue.dat};
\addplot[estmarks] table[x=f,y=qd_ph] {figures/data/case2_Zhat.dat};
\nextgroupplot[ ymin=-30, ymax=60, ytick={-30,0,30,60},
               xlabel={Frequency [\si{\Hz}]}]
\addplot[truecurve] table[x=f,y=qq_ph] {figures/data/case2_Ztrue.dat};
\addplot[estmarks] table[x=f,y=qq_ph] {figures/data/case2_Zhat.dat};
\end{groupplot}
\end{tikzpicture}
\caption{Case~2: magnitude and phase of the four real TFs together with the true values.
The added converter breaks the symmetry: $Z_{dd}\neq Z_{qq}$ and
$Z_{dq}\neq-Z_{qd}$, most visibly below \SI{100}{\Hz}.}
\label{fig:c2Z}
\end{figure*}
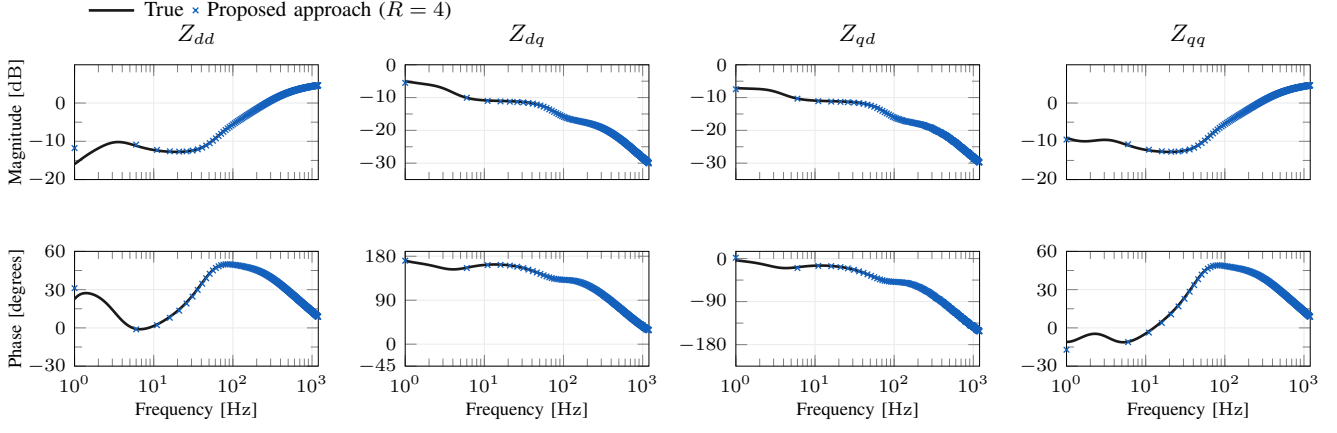

%% file: figures/fig_zoom.tex
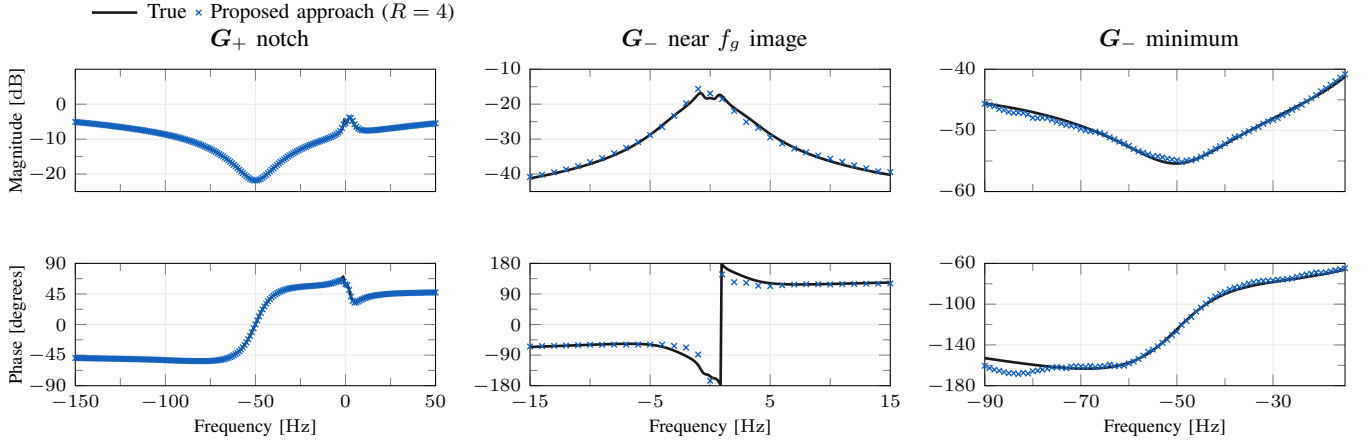
\begin{figure*}[t]
\centering
\begin{tikzpicture}
\begin{groupplot}[group style={group size=3 by 2, horizontal sep=1.25cm, vertical sep=0.95cm},
                   hilaxis, width=0.35\textwidth, height=3.2cm,
                   title style={font=\small, yshift=-0.6ex},
                   legend style={at={(0.02,1.30)}, anchor=south west, legend columns=2}]
\nextgroupplot[ylabel={Magnitude [\si{\dB}]}, xmin=-150, xmax=50, xtick={-150,-100,-50,0,50}, ymin=-25, ymax=10,
               ytick={0,-10,-20}, title={$\mathstrut$ $\bG_+$ notch}, xticklabels={,,}]
\addplot[truecurve] table[x=f,y=Gp_mag] {figures/data/case2_zoomGp_true.dat};
\addplot[estmarks] table[x=f,y=Gp_mag] {figures/data/case2_zoomGp_hat.dat};
\legend{True, Proposed approach ($R=4$)}
\nextgroupplot[ xmin=-15, xmax=15, xtick={-15,-5,5,15}, ymin=-45, ymax=-10,
               ytick={-10,-20,-30,-40}, title={$\mathstrut$ $\bG_-$ near $f_g$ image}, xticklabels={,,}]
\addplot[truecurve] table[x=f,y=Gm_mag] {figures/data/case2_zoom0_true.dat};
\addplot[estmarks] table[x=f,y=Gm_mag] {figures/data/case2_zoom0_hat.dat};
\nextgroupplot[ xmin=-90, xmax=-15, xtick={-90,-70,-50,-30}, ymin=-60, ymax=-40,
               ytick={-40,-50,-60}, title={$\mathstrut$ $\bG_-$ minimum}, xticklabels={,,}]
\addplot[truecurve] table[x=f,y=Gm_mag] {figures/data/case2_zoomM50_true.dat};
\addplot[estmarks] table[x=f,y=Gm_mag] {figures/data/case2_zoomM50_hat.dat};
\nextgroupplot[ylabel={Phase [degrees]}, xmin=-150, xmax=50, xtick={-150,-100,-50,0,50}, ymin=-90, ymax=90,
               ytick={-90,-45,0,45,90}, xlabel={Frequency [\si{\Hz}]}]
\addplot[truecurve] table[x=f,y=Gp_ph] {figures/data/case2_zoomGp_true.dat};
\addplot[estmarks] table[x=f,y=Gp_ph] {figures/data/case2_zoomGp_hat.dat};
\nextgroupplot[ xmin=-15, xmax=15, xtick={-15,-5,5,15}, ymin=-180, ymax=180,
               ytick={-180,-90,0,90,180}, xlabel={Frequency [\si{\Hz}]}]
\addplot[truecurve] table[x=f,y=Gm_ph] {figures/data/case2_zoom0_true.dat};
\addplot[estmarks] table[x=f,y=Gm_ph] {figures/data/case2_zoom0_hat.dat};
\nextgroupplot[ xmin=-90, xmax=-15, xtick={-90,-70,-50,-30}, ymin=-180, ymax=-60,
               ytick={-180,-140,-100,-60}, xlabel={Frequency [\si{\Hz}]}]
\addplot[truecurve] table[x=f,y=Gm_ph] {figures/data/case2_zoomM50_true.dat};
\addplot[estmarks] table[x=f,y=Gm_ph] {figures/data/case2_zoomM50_hat.dat};
\end{groupplot}
\end{tikzpicture}
\caption{Case~2, sharp features resolved at \SI{1}{\Hz}, with every estimated
spectral line shown. Left: the $-\SI{50}{\Hz}$ antiresonance of $\bG_+$,
\SI{26.1}{\dB} deep with a \SI{3}{\dB} width of \SI{17.6}{\Hz}. Center:
the near-fundamental peak of $\bG_-$. Right: the $-\SI{50}{\Hz}$ minimum of
$\bG_-$.}
\label{fig:zoom}
\end{figure*}

%% file: sections/table_metrics.tex
\begin{table}[t]
\renewcommand{\arraystretch}{1.12}
\caption{Identification accuracy, one-second record, $R=4$. Real
responses evaluated over $[1,1200]\,\si{\Hz}$ and complex responses over
$[-1,1]\,\si{\kHz}$.}
\label{tab:metrics}
\centering
\footnotesize
\setlength{\tabcolsep}{3.2pt}
\begin{tabular}{@{}l cccc c@{}}
\toprule
& \multicolumn{4}{c}{Fit\% \eqref{eq:fit}} & Relative \\
\cmidrule(lr){2-5}
 & $Z_{dd}$ & $Z_{dq}$ & $Z_{qd}$ & $Z_{qq}$ & $H_\infty$ error \eqref{eq:hinf}\\
\midrule
Case 1 (symmetric)  & 100.00 & 100.00 & 100.00 & 100.00 & 0.0036 \\
Case 2 (asymmetric) & 99.97  & 99.77  & 99.78  & 99.98  & 0.0654 \\
\midrule
& \multicolumn{2}{c}{Fit\% \eqref{eq:fit}} & \multicolumn{2}{c}{Relative $H_\infty$ error} & \\
\cmidrule(lr){2-3}\cmidrule(lr){4-5}
 & $\bG_+$ & $\bG_-$ & $\bG_+$ & $\bG_-$ & \\
\midrule
Case 1 (symmetric)  & 100.00 & --\textsuperscript{a} & 0.0071 & --\textsuperscript{a} & \\
Case 2 (asymmetric) & 99.99  & 85.57 & 0.0374 & 0.6218 & \\
\bottomrule
\end{tabular}

\vspace{0.4ex}
{\scriptsize \textsuperscript{a}\,Undefined: the true $\bG_-$ is identically
zero, so the normalizations in \eqref{eq:fit} and \eqref{eq:hinf} vanish. The
estimate is reported as an absolute noise floor in the text instead.}
\end{table}

%% file: sections/table_order.tex
\begin{table}[t]
\renewcommand{\arraystretch}{1.12}
\caption{Accuracy against the local model order $R$, with $\ell = 4R+2$, over
$[1,1200]\,\si{\Hz}$.}
\label{tab:order}
\centering
\footnotesize
\setlength{\tabcolsep}{3.6pt}
\begin{tabular}{@{}c l cccc c@{}}
\toprule
& & \multicolumn{4}{c}{Fit\%} & Relative \\
\cmidrule(lr){3-6}
$R$ & Case & $Z_{dd}$ & $Z_{dq}$ & $Z_{qd}$ & $Z_{qq}$ & $H_\infty$ error\\
\midrule
\multirow{2}{*}{4}  & 1 (symmetric)  & 100 & 100 & 100 & 100 & 0.0036 \\
                    & 2 (asymmetric) & 99.974 & 99.767 & 99.779 & 99.978 & 0.0654 \\
\midrule
\multirow{2}{*}{6}  & 1 (symmetric)  & 100 & 100 & 100 & 100 & 0.0036 \\
                    & 2 (asymmetric) & 99.975 & 99.765 & 99.791 & 99.978 & 0.0615 \\
\midrule
\multirow{2}{*}{8}  & 1 (symmetric)  & 100 & 100 & 100 & 100 & 0.0034 \\
                    & 2 (asymmetric) & 99.975 & 99.771 & 99.797 & 99.978 & 0.0631 \\
\midrule
\multirow{2}{*}{10} & 1 (symmetric)  & 100 & 100 & 100 & 100 & 0.0034 \\
                    & 2 (asymmetric) & 99.976 & 99.765 & 99.797 & 99.978 & 0.0571 \\
\bottomrule
\end{tabular}
\end{table}

%% file: sections/table_order_complex.tex
\begin{table}[t]
\renewcommand{\arraystretch}{1.12}
\caption{Accuracy of the complex TF estimates against the local model order $R$,
with $\ell = 4R+2$, over $[-1,1]\,\si{\kHz}$.}
\label{tab:ordercomplex}
\centering
\footnotesize
\setlength{\tabcolsep}{4pt}
\begin{tabular}{@{}c l cc cc@{}}
\toprule
& & \multicolumn{2}{c}{$\bG_+$} & \multicolumn{2}{c}{$\bG_-$}\\
\cmidrule(lr){3-4}\cmidrule(lr){5-6}
$R$ & Case & Fit\% & Rel.\ $H_\infty$ & Fit\% & Rel.\ $H_\infty$\\
\midrule
\multirow{2}{*}{4}  & 1 (symmetric)  & 100 & 0.0071 & --\textsuperscript{a} & --\textsuperscript{a} \\
                    & 2 (asymmetric) & 99.988 & 0.0374 & 85.57 & 0.6218 \\
\midrule
\multirow{2}{*}{6}  & 1 (symmetric)  & 100 & 0.0066 & --\textsuperscript{a} & --\textsuperscript{a} \\
                    & 2 (asymmetric) & 99.988 & 0.0340 & 87.26 & 0.5933 \\
\midrule
\multirow{2}{*}{8}  & 1 (symmetric)  & 100 & 0.0064 & --\textsuperscript{a} & --\textsuperscript{a} \\
                    & 2 (asymmetric) & 99.988 & 0.0346 & 88.23 & 0.5928 \\
\midrule
\multirow{2}{*}{10} & 1 (symmetric)  & 100 & 0.0046 & --\textsuperscript{a} & --\textsuperscript{a} \\
                    & 2 (asymmetric) & 99.988 & 0.0308 & 88.94 & 0.5778 \\
\bottomrule
\end{tabular}

\vspace{0.4ex}
{\scriptsize \textsuperscript{a}\,Undefined: the true $\bG_-$ is identically
zero in Case~1, so the normalizations in \eqref{eq:fit} and \eqref{eq:hinf}
vanish.}
\end{table}

%% file: sections/conclusions.tex
\section{Conclusions}\label{sec:conclusions}

The paper developed a non-parametric frequency-domain method for
identifying $dq$-asymmetric grid impedances from a {single arbitrary record}. The equivalent impedance is parameterized by a pair of SISO
complex transfer functions, and each spectral line is fitted with a local
rational model in which the transient and leakage contributions are estimated
jointly with the responses rather than suppressed by windowing. This removes
both the requirement for periodic steady-state data and the requirement for
sequential perturbation, leaving the record length governed only by the desired
frequency resolution and by the SNR.

The exact finite-time DFT relation of the measured currents and voltages was proved for the conjugate-coupled complex-signal case. It allows the transient term to be estimated together with the system responses. The identifiability and excitation conditions for the proposed method were also established. It was also shown that a 
stationary frame dynamical filter (e.g., anti-alias/decimation filter)  distorts the estimate of
$\Gm$, but not $\Gp$. This occurs even when both use identical transducers and filters. The method was validated on a  HIL platform against an analytically
derived small-signal model, for a passive $dq$-symmetric grid and for the same
grid with an added grid-following converter that renders it asymmetric. The
responses of both complex TFs and of the four real TFs were recovered over a
wide band from a single one-second record at \SI{1}{\Hz} frequency resolution.

%% file: sections/appendix.tex
\section{Conjugation and Reversal of the DFT}\label{app:conj}
\vspace{-0.15cm}
Let $\{\bi(t_n)\}$, $n \in \{0,\dots,N-1\}$, be a complex-valued record with
$N$-point DFT $\bI\!_k$ given by the definition in \eqref{eq:dft_spectrum}. The DFT of the conjugated record
$\{\bi^\ast(t_n)\}$ at the line $k$ is
\[
\begin{aligned}
\frac{1}{\sqrt{N}} \sum_{n=0}^{N-1} \bi^\ast(t_n)\, & \eu^{-j 2\pi k n /N}
 = \left[\frac{1}{\sqrt{N}} \sum_{n=0}^{N-1} \bi(t_n)\,
    \eu^{+j 2\pi k n /N}\right]^{\!\ast}\\[0.2em]
 & = \left[\frac{1}{\sqrt{N}} \sum_{n=0}^{N-1} \bi(t_n)\,
    \eu^{-j 2\pi (-k) n /N}\right]^{\!\ast}
  = \bI^{\!\ast}_{-k},
\end{aligned}
\]
where $\bI\!_{-k}$ is understood as \eqref{eq:dft_spectrum}
evaluated at the index $-k$. Since $\eu^{-j 2\pi (k+N) n /N} = \eu^{-j 2\pi k n /N}$
for every integer $n$, that sum is $N$-periodic in $k$, so
$\bI\!_{-k} = \bI\!_{N-k}$, and folding the index into $\{0,\dots,N-1\}$ gives
\begin{equation}\label{eq:conjrev}
\frac{1}{\sqrt{N}} \sum_{n=0}^{N-1} \bi^\ast(t_n)\, \eu^{-j 2\pi k n /N}
= \bI^{\!\ast}_{(N-k)_N},
\end{equation}
with  $ k \in \{0,\dots,N-1\}$. Notice that $(N-0)_N = 0$ and \eqref{eq:conjrev}
reduces to $\bI^{\!\ast}_0$: the line at $k=0$ is its own mirror under the index
reversal, though for a complex record $\bI_0$ is not real.

If the signal is real, then $\bi^\ast = \bi$, the left
side of \eqref{eq:conjrev} is $\bI\!_k$ itself, and \eqref{eq:conjrev} becomes
the Hermitian symmetry $\bI\!_k = \bI^{\!\ast}_{(N-k)_N}$, which halves the
number of independent lines. For the complex records considered here the two
spectra are distinct.

\section{Proof of Theorem~\ref{thm:dft}}\label{app:dft}
\vspace{-0.15cm}
Define the Fourier and convolution integrals
\begingroup\makeatletter\def\f@size{9}\check@mathfonts
\[
\F_a^b\{\bv\}(\omega) \!:=\!\!\int_a^b\! \bv(t) \eu^{-j\omega t} \diff t, \quad
(\bg\ast \bi)_a^b(t)  \!:=\!\! \int_a^b\! \bg(t-\tau) \bi(\tau) \diff \tau,
\]
\endgroup
and drop the arguments $\omega$ and $t$ where no confusion arises. The
$dq$-symmetric case is treated first; the general case follows by applying it to
the responses of $\bG_+$ and $\bG_-$ separately.\vspace{-0.3cm}

\subsection{The symmetric case}

Write $\bv(t) = (\bg\ast \bi)_0^t  + (\bg\ast \bi)_{-\infty}^0$, where $\bg$ is
the impulse response of the stable and \textit{causal} $\bG$ and $t\in[0,T]$
with $T=NT_s$. Then
\[
\begin{aligned}
    \bV(\omega)   &:= \F_{-\infty}^\infty\{\shat \bv\} = \F_0^T\{\bv\}  \\
                  &=  \F_0^T\{(\bg\ast \bi)_0^t\} +  \F_0^T\{(\bg\ast \bi)_{-\infty}^0\} \\
                  & =  \F_0^T\{(\bg\ast \bi)_0^\infty\} -\underbrace{\F_0^T\{(\bg\ast \bi)_t^\infty\}}_{= \, 0 \text{ by causality of $\bG$}} +  \F_0^T\{(\bg\ast \bi)_{-\infty}^0\} \\[0.3em]
                  &=   \F_0^\infty\{(\bg\ast \bi)_0^\infty\} -\F_T^\infty\{(\bg\ast \bi)_0^\infty\} +\F_0^T\{(\bg\ast \bi)_{-\infty}^0\}\\
                  &= \bG(j\omega) \bI(\omega) - \bV_{\!\text{fin}}(\omega) + \bV_{\!\text{init}}(\omega).
\end{aligned}
\]
Here $\bV(\omega)$ and $\bI(\omega)$ are the spectra of the complex
continuous-time signals $\shat \bv$ and $\shat\bi$, defined as equal to $\bv$
and $\bi$ on $[0,T]$ and zero elsewhere. In the third line, the second term
vanishes for causal $\bG$, since then $\bg(t) = 0$ for all $t<0$. The last
equality uses the convolution theorem,
$\F_0^\infty\{(\bg\ast \bi)_0^\infty\} = \bG(j\omega)\F_0^\infty\{\bi(t)\}$,
together with $\bi(t) = 0$ for $t>T$, a benign assumption because $\bG$ is
causal. The leakage terms are
$\bV_{\!\text{fin}}(\omega):= \F_T^\infty\{(\bg\ast \bi)_0^\infty\}$ and
$\bV_{\!\text{init}}(\omega) :=\F_0^T\{(\bg\ast \bi)_{-\infty}^0\}$,   corresponding to the final and initial conditions, respectively.

It remains to connect the continuous spectra to the discrete ones in
\eqref{eq:dft_spectrum}.  Sampling of continuous functions is represented mathematically as the convolution of the continuous function with impulses $\delta(t_n)$ defined at the sampling point $t_n$. Expanding the impulse train $\sum_n \delta(t+nT)$ in a
Fourier series with $\omega_1 := 2\pi/T$, convolving with $\shat\bv(t)$ and
invoking the Fourier symmetry theorem yields Poisson's summation formula
\cite[pg.~75, (3-87) \& Note 2]{papoulis1977signal},
\[
\sum_{n = -\infty}^\infty \bV(\omega+n\omega_s) = \frac{2\pi}{\omega_s} \sum_{n=0}^{N} {c(n)}{\bv(nT_s)} \eu^{-jnT_s\omega},
\]
in which the sum on the right is finite because $\bV$ is a finite-time Fourier
integral over $[0,T]$, and the weight $c(n) = \tfrac12$ for $n\in\{0,N\}$ and $1$
otherwise accounts for the discontinuity of $\shat\bv$ at the window end points.
Sampling the spectra uniformly with resolution $2\pi/T$ at $\omega_k :=2\pi k/T$
and rearranging using \eqref{eq:dft_spectrum} gives,
\[
\bV\!_k = \frac{1}{\sqrt{N} T_s}\sum_{n = -\infty}^\infty \bV\!\left(\omega_k - n \omega_s\right) + \frac{\bv(0) - \bv(T)}{2\sqrt{N}},
\]
and the same relation holds for $\bI_k, \bV_{\!\text{init},k}$ and
$\bV_{\!\text{fin},k}$. Collecting terms gives
$\bV\!_k = \bG(j\omega_k) \bI_k + \bT(j\omega_k)$ with
$\bT(j\omega_k) = \bV_{\!\text{init},k} - \bV_{\!\text{fin},k} + \boldsymbol{\alpha}_k$,
where $\boldsymbol{\alpha}_k$ carries the aliasing.

\subsection{The asymmetric case}

Split \eqref{eq:complex_model} as $\bv(t) = \bv_+(t) + \bv_-(t)$ with
\[
\begin{aligned}
\bv_+(t) &= (\bg_+\ast \bi)_0^t  \;\,+ (\bg_+\ast \bi)_{-\infty}^0,\\
\bv_-(t) &= (\bg_-\ast \bi^\ast)_0^t  + (\bg_-\ast \bi^\ast)_{-\infty}^0,
\end{aligned}
\]
where $\bg_+$ and $\bg_-$ are the impulse responses of the causal $\bG_+$ and
$\bG_-$. The symmetric-case result applies to each, giving
\[
\begin{aligned}
\bV\!_{+k} &= \bG_+(j\omega_k) \bI_k + \bT_+(j\omega_k),\\
\bV\!_{-k} &= \bG_-(j\omega_k) \bI_{(N-k)_N}^\ast + \bT_-(j\omega_k).
\end{aligned}
\]
The reversed and conjugated spectrum in the second line arises because
$\F_0^\infty\{(\bg\ast \bi^\ast)_0^\infty\} = \bG(j\omega)[\bI(-\omega)]^\ast$,
which follows from
$\bI^\ast(\omega) := \F_0^\infty\{\bi^\ast(t)\} = [\bI(-\omega)]^\ast$; on the
sampled grid $-\omega_k$ is the line $(N-k)_N$ by
Appendix~\ref{app:conj}. Adding the two gives
\eqref{eq:dft_relation} with $\bT = \bT_+ + \bT_-$.

The transient term carries two effects: the leakage in
$\bV_{\!\text{init},k}$ and $\bV_{\!\text{fin},k}$, unforced decaying responses
due to the initial and final conditions, therefore governed by the same dynamics as $\bG_+$ and $\bG_-$;
and the aliasing in $\boldsymbol{\alpha}_k$, which carries an infinite
repetition of those poles through the folding of the spectrum. Both
are smooth over a short frequency interval; therefore, $\bT$ can be approximated by a low-order local rational model alongside the responses.
The $\mathcal{O}(N^{-1/2})$ decay follows from the $1/\sqrt{N}$ normalization in
\eqref{eq:dft_spectrum} and the boundedness of the end-point terms.

\section{Rank Condition on \texorpdfstring{$\Phi_k$}{Phi k}}\label{app:rank}
\vspace{-0.15cm}
The condition stated in Section~\ref{sec:excitation} is the following.

\begin{lemma}\label{lem:excitation}
Let \eqref{eq:count} hold. If $\Phi_k$ is rank deficient, then there exist
complex polynomials $\alpha, \beta^+, \beta^-, \gamma$ of degree at most $R$,
with $\alpha(0) = 0$ and not all zero, such that
\begin{equation}\label{eq:degeneracy}
\alpha(r)\bV\!_{k+r} = \beta^+(r) \bI_{k+r} + \beta^-(r) \bI^\ast_{(N-k-r)_N} + \gamma(r)
\end{equation}
for every $r \in \{-\ell,\dots,\ell\}$. Conversely, if the local current
spectrum together with its mirror satisfies \eqref{eq:ratform} for polynomials
of degree at most $R$ not all zero, then $\Phi_k$ is rank deficient and the
solution is not unique.
\end{lemma}

\begin{IEEEproof}
The proof is based on linear algebra arguments. Let $\Phi_k \theta = 0$ for some
$\theta \neq 0$, and collect its entries into the polynomials $\alpha(r)$,
$\beta^\pm(r)$ and $\gamma(r)$ of degree at most $R$. Here $\alpha(0) := 0$
because the first block of \eqref{eq:datamatrix} uses $\widetilde{\Phi}$. Reading the row of $\Phi_k$ indexed by $r$ then gives
\eqref{eq:degeneracy}. Conversely, \eqref{eq:ratform} is the special case
$\alpha \equiv 0$, for which the corresponding $\theta$ is non-zero by
construction, so the columns of $\Phi_k$ are linearly dependent.
\end{IEEEproof}

Substituting \eqref{eq:dft_relation} into \eqref{eq:degeneracy} shows what the
condition of the Lemma amounts to. The degeneracy is equivalent to
\[
\bI_{k+r}\big[\beta^+ - \alpha \bG_+\big] + \bI^\ast_{(N-k-r)_N}\big[\beta^- - \alpha \bG_-\big] + \big[\gamma - \alpha \bT\big] = 0
\]
over the interval, that is, to a polynomial relation of degree at most $R$
between the local current spectrum, its mirror, and a constant sequence.

\section{The DFT Line at Zero Frequency}\label{app:dcbin}

The following simple  example shows why the line at $\omega = 0$ has to be removed
from the local interval. Take $N = 8$ spectral lines, $R = 1$ and $\ell = 3$, so
that $r \in \{-3,\dots,3\}$ and the interval carries $2\ell+1 = 7$ data
equations for $4R+3 = 7$ unknowns, and center it on $k = 0$. Using the
periodicity of the DFT, the direct and mirrored indices over the interval are
\[
\begin{gathered}
k+r \;\rightarrow\; 5,\,6,\,7,\,0,\,1,\,2,\,3,\\
(N-k-r)_N \;\rightarrow\; 3,\,2,\,1,\,0,\,7,\,6,\,5 .
\end{gathered}
\]
With $\theta_0 = [\,a_1\;\; b_0^+\;\; b_1^+\;\; b_0^-\;\; b_1^-\;\; c_0\;\; c_1\,]^\top$,
the local model $\bY\!_0 = \Phi_0\theta_0$ reads
\begin{equation*}
\resizebox{\columnwidth}{!}{$
\begin{bmatrix}
\bV\!_5\\ \bV\!_6\\ \bV\!_7\\ \bV\!_0\\ \bV\!_1\\ \bV\!_2\\ \bV\!_3
\end{bmatrix}
=
\begin{bmatrix}
 3\bV\!_5 & \bI_5 & -3\bI_5 & \bI^\ast_3 & -3\bI^\ast_3 & 1 & -3\\
 2\bV\!_6 & \bI_6 & -2\bI_6 & \bI^\ast_2 & -2\bI^\ast_2 & 1 & -2\\
  \bV\!_7 & \bI_7 &  -\bI_7 & \bI^\ast_1 &  -\bI^\ast_1 & 1 & -1\\
        0 & \bI_0 &       0 & \bI^\ast_0 &            0 & 1 &  0\\
 -\bV\!_1 & \bI_1 &   \bI_1 & \bI^\ast_7 &   \bI^\ast_7 & 1 &  1\\
 -2\bV\!_2 & \bI_2 &  2\bI_2 & \bI^\ast_6 &  2\bI^\ast_6 & 1 &  2\\
 -3\bV\!_3 & \bI_3 &  3\bI_3 & \bI^\ast_5 &  3\bI^\ast_5 & 1 &  3
\end{bmatrix}
\theta_0
$}
\end{equation*}
The fourth row is the one at $r = 0$, that is, at the line $k = 0$. Every entry
weighted by $r$ vanishes there.  The row therefore reduces to
\[
\bV\!_0 = b_0^+ \bI_0 + b_0^- \bI_0^\ast + c_0.
\]
Moreover, the small-signal records have their mean removed, as
described in Section~\ref{sec:formulation}, so $\bV\!_0 = \bI_0 = 0$ and the row
becomes $c_0 = 0$. The constant coefficient of $\bC_k$ is thus forced to zero
irrespective of the data, and the estimate of the leakage term is biased in every local
interval that contains this line. Removing the row removes this constraint. The zero 
frequency simply counts as one spectral line, and with the interval widths used
here its loss is immaterial.

\section{Proof of Proposition~\ref{prop:image}}\label{app:distortion}
Write $\tilde{\bV}\!_k$ and $\tilde{\bI}_k$ for the DFT spectra of the records
that the estimator actually sees, that is, after the acquisition chains and the
Park transform. Then, it holds that
$\tilde{\bV}\!_k = C^{v,+}_{k}\bV\!_k + \bT^{v,+}_{k}$ and $\tilde{\bI}_k = C^{i,+}_{k}\bI_k+\bT^{i,+}_{k}$, where the second terms on the right-hand side of the equations represent the leakage.
The conjugate regressor is taken at the mirrored line $(N-k)_N$, whose $dq$
frequency is $-\omega_k$, so the current chain is evaluated at
$\omega_g-\omega_k$ there and
$\tilde{\bI}^\ast_{(N-k)_N} = [{C^{i,-}_{k}}]^\ast\,\bI^\ast_{(N-k)_N}  + [\bT^{i,-}_{k}]^\ast$. %
From this, and using \eqref{eq:dft_relation}
\[
\begin{aligned}
\tilde{\bV}\!_k &= \frac{C^{v,+}_{k}}{C^{i,+}_{k}}\Gp(j\omega_k)\,\tilde{\bI}_k
 + \frac{C^{v,+}_{k}}{[{C^{i,-}_{k}}]^\ast}\Gm(j\omega_k)\,\tilde{\bI}^\ast_{(N-k)_N}  ,\\
 \tilde{\bT}_k &= C^{v,+}_{k}\bT(j\omega_k)+ \bT^{v,+}_{k} \\
 & \qquad\quad -\frac{C^{v,+}_{k}}{C^{i,+}_{k}} \Gp(j\omega_k) \bT^{i,+}_{k} - \frac{C^{v,+}_{k}}{[{C^{i,-}_{k}}]^\ast}  \Gm(j\omega_k) [\bT^{i,-}_{k}]^\ast
\end{aligned}
\]
Wherever the factors are non-zero and vary smoothly over each local interval,
which holds for FIR and transducer responses across the band, the
reparameterization stays within the local model class and preserves the linear
independence of the two input columns of \eqref{eq:datamatrix}.